\documentclass[11pt,a4paper,english]{article}

 \title{Strict Almost Non-Abelian Deformations}
 \author{A. Much\\ \footnotesize{Instituto de Ciencias Nucleares, UNAM, M\'exico D.F. 04510, M\'exico}
 	\\ \footnotesize{Faculty of Mathematics, University of Vienna, 1090 Vienna, Austria}}

 \usepackage{float}            % bietet Option [H] fÃƒÂ¼r bombenfestes Verankern
 \usepackage[T1]{fontenc}   
 \usepackage[cmex10]{amsmath}  % fÃƒÂ¼r erw. Formeloptionen, Option [] zur Vermeidung von Type3-Fonts
 \usepackage{amstext}          % fÃƒÂ¼r Klartext via \text{} in Formeln
 \usepackage{mathrsfs}
 \usepackage{amsfonts}         % fÃƒÂ¼r komplexere Formeln (Mengensymbole ...)
 \usepackage{amssymb}          % fÃƒÂ¼r komplexere Formeln (Mengensymbole ...)
 \usepackage{amsthm}    
 \usepackage{setspace}
 \usepackage{bm}               % bold math, fÃƒÂ¼r \bm{}
 \usepackage{enumerate}        % verbessert AufzÃƒÂ¤hlungen
 \usepackage[bottom]{footmisc} % Fussnoten am Seitenende
 \usepackage{array}            % fÃƒÂ¼r Tabellen: bindet tabular-Umgebung ein
 \usepackage{textcomp}         % fÃƒÂ¼r \textdegree , \textcelsius , macht aber manchmal Probleme (!!)
 \usepackage{pdfpages}         % fÃƒÂ¼r die Einbindung kompletter pdf-*Seiten*
 \usepackage{parskip}          % zw. AbsÃƒÂ¤tzen: eine knappe Leerzeile statt hÃƒÂ¤ngender EinzÃƒÂ¼ge
 \usepackage[right]{eurosym}   % Eurosymbol
 \usepackage{xcolor}           % farbiger Text
 \usepackage[square,numbers]{natbib}	  % FÃƒÂ¼r \setlength{\bibsep}{3mm}; square macht eckige

 \usepackage{amsmath}  
 \usepackage[colorlinks,pdfpagelabels,pdfstartview = FitH,bookmarksopen = true,bookmarksnumbered =
 true,linkcolor = black,plainpages = false,hypertexnames = false,citecolor = black]{hyperref}

 \newtheorem{theorem}{\textsc{Theorem}}[section]
 \newtheorem{lemma}{\textsc{Lemma}}[section]
 \newtheorem{proposition}{\textsc{Proposition}}[section]
 
 \theoremstyle{definition}
 \newtheorem{definition}{\textsc{Definition}}[section]
 \newtheorem{Notation}{Notation}[section]
   \newtheorem{convention}{Conventions}[section]
 
 \theoremstyle{remark}
 \newtheorem{remark}{Remark}[section]

 \usepackage{fancyhdr}
 \numberwithin{equation}{section} 
 
\begin{document}
 	\maketitle
 	\abstract{An almost non-abelian extension of the Rieffel deformation is presented in this work. The non-abelicity comes into play by the introduction of   unitary groups which are dependent of the infinitesimal generators of $SU(n)$. This extension is applied to quantum mechanics and quantum field theory.} 
 	%%%%%%%%%%%%%%%%%%%%%%
 	% Inhaltsverzeichnis %
 	\tableofcontents
 	%%%%%%%%%%%%%%%%%%%%%%
 	
 	\section{Introduction}
 	Is it possible to understand and investigate the weak and  strong interaction in terms of a deformation method?  The intuition leading to this question comes from using warped convolutions \cite{GL1, BS, BLS} in the context of quantum measurement theory and quantum mechanics (QM) \cite{AA, Muc1}.
 	By using the novel tool of warped convolutions in a quantum mechanical context, we recast  many fundamental physical effects involving electromagnetism. In particular, the deformation of a free Hamiltonian becomes, after setting deformation parameters to physical constants, the minimally coupled Hamiltonian. This means that through deformation we generated abelian gauge fields. Hence, the main question  following this insight goes as follows: Can we formulate a similar rigorous apparatus in order to obtain non-abelian gauge fields by deformation? 
 	\newline\newline
 	The investigation shows that there is a natural extension to the framework of warped convolutions in the non-abelian context. In particular, the formulation in the non-abelian case resembles the abelian version in many parts. Most of the theorems, lemmas and propositions that were stated in the abelian case hold in the non-abelian version as well. For example, we can express the deformation in terms of strong limits, which is equivalent to the spectral representation. Furthermore, for the operators belonging to a suitable sub-algebra such deformations are well defined. Another example of similarity is the case of symmetry. If a self-adjoint operator is deformed the result will be a symmetric operator. Now if the operator fulfills in addition certain decay requirements, self-adjointness follows.  
 	\newline\newline
 	By applying the apparatus to QM we obtain minimal substitutions corresponding to non-abelian gauge fields. However,  the non-abelian part of the field strength tensor (calculated by the algebra see \cite{Muc1}), i.e. the quadratic terms of the gauge fields vanish.   In particular, the reason for  this  vanishing can be traced back to the necessity of a strongly continuous unitary group in order to fulfill the requirements needed for warped convolutions. Nevertheless, following ideas outlined in the abelian deformation of QM, we are able to define a non-commutative   space-time, which we  name the non-abelian Moyal-Weyl. 
 	\newline\newline 
 	By using the algebra of the  new space-time, we define a QFT on the four dimensional  non-abelian Moyal-Weyl. 
 	This is done in order to obtain or construct
 	non-trivial interacting fields which was done for the abelian case in \cite{A, GL1, GL2, GL4, GL5, Mor, MUc}. We further prove the isomorphism to   fields defined on the  non-abelian Moyal-Weyl to quantum fields (QF) deformed with non-abelian-warped convolutions. This is done by explicitly constructing a unitary operator. Furthermore, we prove the Wightman properties of the deformed QF and relate them to fields defined on a wedge.
 	\newline\newline 
 	The advantage of QFT on wedges is its close relation to a modified version of Lorentz-covariance which, under certain requirements, lead to local observables. It was recognized by the authors in  \cite{GL1} that QF on a Moyal-Weyl space-time can be identified with fields on the wedge and it is rewarded with an exact two particle scattering matrix.    The notion of wedge-covariance and wedge-locality (the vanishing of the commutator of two fields which are wedge-like separated) can be applied as well in our case. 
 	\newline\newline 
 	The paper is organized as follows: The second section   introduces some important definitions and lemmas regarding the original warped convolutions. The third section is devoted to the mathematical formulation of non-abelian deformations. In the fourth and fifth  sections we apply non-abelian-deformations to QM and QFT. 
 		
 		\begin{convention}
 			Throughout this work we use $d=n+1$, for $n\in\mathbb{N}$. The Greek letters are split into  $\mu, \,\nu=0,\dots,n$ and for the $su(m)$-generators we use $\alpha, \beta,\gamma=1,\cdots,m^2-1$. Moreover, we use Latin letters for the spatial components which run from $1,\dots,n$ and we choose the following convention for the Minkowski scalar product of $d $-dimensional vectors, $a\cdot b=a_0b^0+a_kb^k=a_0b^0- \vec{a}\cdot\vec{b}$.
 		\end{convention}
 		
 		\section{Warped convolutions}\label{s2}
 		In the current work, warped convolutions are generalized to an almost non-abelian case.  Since we permanently refer to the method  we lay out in this section the novel deformation procedure
 		and present the most important definitions, lemmas and propositions for the current work. 
 		\newline\newline 
 		We start by assuming the existence of a strongly continuous unitary group $U$ that is a representation of the additive
 		group $\mathbb{R}^{d}$ on some separable Hilbert space $\mathcal{H}$. 
 		\begin{Notation}
 			We denote by $\mathcal{C}\subset\mathcal{B}(\mathcal{H})$ the $C^*$-algebra of all bounded operators on $\mathcal{H}$ on which the $\mathbb{R}^d$-action
 			$$
 			\alpha_k(A):=U(k)\, A\,U(k)^{-1}
 			$$
 			is strongly continuous, i.e. such that the function $\mathbb{R}^d\ni k\to\|\alpha_k(A) \|$ is continuous. We denote by $\mathcal{C}^\infty$ the $*$-subalgebra of $\alpha$-smooth elements in $\mathcal{C}$ and  let $E$ be the spectral resolution of the unitary operator $U$, i.e.
 			$$
 			U(k)=\int_{\mathbb{R}^d}e^{ikx}\, dE(x).
 			$$
 			Finally,   let $\mathcal{D}$ be the dense and  stable subspace of vectors in $\mathcal{H}$ which transform smoothly under $U$. By using the former notations and definitions we are now able to give the definition of the warped convolutions of an operator.
 		\end{Notation}
 		\begin{definition}\label{defwca}
 			Let $\Theta$ be a real skew-symmetric matrix, regarded as an operator on $\mathbb{R}^d$. The warped convolutions  of an operator $A\in\mathcal{C}^\infty$ with respect to $(\alpha,\Theta)$ are the operators $A_{\Theta}$ and  $_{\Theta}A$ defined on the domain $\mathcal{D}\subset\mathcal{H}$ according to
 			\begin{equation}\label{WC1}
 			A_{\Theta}:=\int dE(x)\,\alpha_{\Theta x}(A) ,\qquad\qquad _{\Theta}A:=\int \alpha_{\Theta x}(A)\,dE(x) .
 			\end{equation}
 		\end{definition}
 		The restriction to smooth operators could perhaps be lifted slightly, but since the deformation is preformed with operator-valued integrals, it is difficult to ensure that the formula makes sense. In \cite{BLS}, it is proven that one can represent the warped
 		convolution $A_\Theta$ and $_{\Theta}A$ of $A \in \mathcal{C}^{\infty}$, on the dense domain
 		$\mathcal{D}\subset\mathcal{H}$ of vectors smooth w.r.t. the action of $U$,  in terms of strong limits 
 		\begin{align*}
 		\int dE(x)\,\alpha_{\Theta x}(A) \Phi=(2\pi)^{-d}
 		\lim_{\epsilon\rightarrow 0}
 		\int\int dx\, dy \,\chi(\epsilon x,\epsilon y )\,e^{-ixy}\,U(y)\, \alpha_{\Theta x}(A) \Phi,  \\
 			\int \alpha_{\Theta x}(A)\, dE(x)\Phi=(2\pi)^{-d}
 			\lim_{\epsilon\rightarrow 0}
 			\int\int dx\, dy \,\chi(\epsilon x,\epsilon y )\,e^{-ixy}\, \alpha_{\Theta x}(A)\, U(y)\Phi, 
 		\end{align*}
 		where $\chi \in\mathscr{S}(\mathbb{R}^d\times\mathbb{R}^d)$ with $\chi(0,0)=1$. Furthermore, it was proven rigorously, by expressing the warped deformations in terms of strong limits, that the two different deformations are equivalent,
 			\begin{equation}\label{eq1}
 			\int dE(x)\,\alpha_{\Theta x}(A)
 =\int \alpha_{\Theta x}(A)\, dE(x).
 			\end{equation} 
In  this work we  consider unbounded real vector-valued functions of the self-adjoint   operators with domain $\mathcal{E}\subseteq\mathscr{S}(\mathbb{R}^d)$ (the Schwartz space). For this  purpose we recall the following theorem \cite[ Theorem VIII.6]{ RS1}.
 		\begin{theorem}\label{tsa}
 			Let $X$ be an  unbounded vector-valued  self-adjoint   operator,  ${Z}(.)$ be an unbounded real vector-valued Borel function on  $\mathbb{R}^{d}$,   and let the dense domain $\mathcal{D}_{{Z}}$ be given as, 
 			\begin{equation*}
 			\mathcal{D}_{{Z}}=\Big\{\phi\Big|\  \int\limits_{-\infty}^{ \infty} |Z_{\mu}({x})|^2\,d(\phi,P_{x}\phi) <\infty\Big\},
 			\end{equation*}
 			where $\{P_x\}_{x\in\mathbb{R}^d}$ are projection-valued measures on $\mathcal{H}$. Then ${Z}({X})$ is a self-adjoint operator with domain $\mathcal{D}_{{Z}}$. 
 		\end{theorem} 
 		\begin{remark}
 			In the case that ${X}$ is the coordinate or momentum operator, which is used in the QM and QFT sections, the density and stability of domain $\mathcal{D}_{Z}$ can be seen in \cite[Lemma 3.1]{Muc1}. 
 		\end{remark}

 		\section{Nonabelian Deformation}
 		In this Section we define the almost non-abelian deformation. The form of   deformation that resembles in a familiar manner the formula of warped convolutions was deduced from  minimal coupling. To be more precise, a reverse engineering path led us to the form of  the non-abelian formula.  Let us first give the following definition. 
 		\begin{definition}\label{nao}
 			Let $\tau_{\alpha}$ be the infinitesimal generators of the special unitary group $SU(m)$, with commutator relations
 			\begin{equation}
 			[\tau_{\alpha},\tau_{\beta}]= if_{\alpha\beta\gamma}\tau^{\gamma}\qquad  f_{\alpha\beta\gamma}\in\mathbb{C}, \quad \alpha, \beta,\gamma=1,\cdots,m^2-1.
 			\end{equation}
 			Then, we define a matrix valued operator $Q({X}) $ on the Hilbert-space  $\mathcal{D}_{ {Z}}\otimes  {\mathbb{C}^{m}}$  as follows
 			\begin{equation}
 			Q( {X})_{\mu} := Z( {X})_{\mu}\otimes Y^{\alpha} \tau_{\alpha} ,\qquad \mu=0,1,\cdots,n
 			\end{equation}  
 			with   operator valued vector  $Z$  which is a real vector-valued function of a self-adjoint operator (see Theorem \ref{tsa}) and $Y\in\mathbb{R}^{m^2-1}$. Another possible choice for $Y$ is to choose it as  $m^2-1$ complex-valued $m\times m$-matrices such that $Y^{\alpha}\tau_{\alpha}$ is hermitian. However the choice of $Y$ may be chosen, we demand that the matrix  $Y^{\alpha}\tau_{\alpha}$ has no zero eigenvalues. The case of all eigenvalues being zero is given by $Y=0$, which in a sense corresponds to the commutative case.
 		\end{definition}
 			\begin{remark}
 				Note that a similar choice of   $Y$ is given   in \cite{CCM} where they define the coordinate operator of a sphere in order to describe geometry by pure algebraic terms, i.e. $Y=Y^{\alpha}\otimes\tau_{\alpha}$.
 			\end{remark}
 		In the next proposition we discuss the properties of the operator which was defined as the  product of   unbounded real vector-valued Borel functions of the coordinate operator and the infinitesimal generators of the special unitary group. 
 		\begin{proposition}
 			The operators $Q( {X})_0,\dots,Q({X})_n$ as given in Definition \ref{nao} are mutually commuting self-adjoint operators on the domain $\mathcal{D}_{{Z}}\otimes  {\mathbb{C}^{m}}$. Hence, there is a strongly continuous $n$-parameter group of unitaries $U^{\tau}(p)$ on $\mathcal{H}\otimes  {\mathbb{C}^{m}}$ with infinitesimal generator ${Z}({X})\otimes Y^{\alpha}\tau_{\alpha}$:
 			\begin{equation}
 			U^{\tau}(p): =\mathrm{exp}(i p^{\mu} Z_{{\mu}}( {X})\otimes Y^{\alpha}\tau_{\alpha}),\qquad\forall p\in\mathbb{R}^{d}.
 			\end{equation} 
 			
 		\end{proposition}
 		\begin{proof}
 			Self-adjointness is    proven by using Theorem \ref{tsa} and the fact that   infinitesimal generators of the special unitary group are represented as traceless finite-dimensional self-adjoint matrices.  Commutativity can be proven by a simple calculation,
 			\begin{align*}
 			[Q( {X})_{\mu} ,Q( {X})_{\nu}] &=   [Z( {X})_{\mu}\otimes Y^{\alpha}\tau_{\alpha},Z( {X})_{\nu}\otimes Y^{\beta}\tau_{\beta}] \\&= Z( {X})_{\mu} Z( {X})_{\nu}\otimes Y^{\alpha}\tau_{\alpha} Y^{\beta}\tau_{\beta}-Z( {X})_{\nu} Z( {X})_{\mu} \otimes Y^{\beta}\tau_{\beta}Y^{\alpha}\tau_{\alpha}\\&= Z( {X})_{\mu}Z( {X})_{\nu}\otimes Y^{\alpha}\tau_{\alpha} Y^{\beta}\tau_{\beta}-Z( {X})_{\nu}Z( {X})_{\mu}\otimes Y^{\alpha}\tau_{\alpha} Y^{\beta}\tau_{\beta}\\&=[ Z( {X})_{\mu},Z( {X})_{\nu}]\otimes Y^{\alpha}\tau_{\alpha} Y^{\beta}\tau_{\beta}
 			\\&=0.
 			\end{align*} 
 			Continuity and unitarity of the group is due to Stone's theorem \cite[ Theorem VIII.7]{ RS1}.
 		\end{proof}Concerning the forthcoming proofs  the following lemma is essential.

 		\begin{lemma}\label{eis}Let
 		$\mathcal{D}_{ {Q}} \subset\mathcal{H}\otimes \mathbb{C}^{m}$
 		be a dense and stable subspace of vectors which are smooth w.r.t. the action of $U^{\tau}$	and let the hermitian $m\times m$ matrix $Y^{\alpha}\tau_{\alpha}$ be written as follows,
 			\begin{equation}
 			Y^{\alpha}\tau_{\alpha}= W\left(\sum_{r=1}^{m} \lambda_r  \, B_r\right) W^{-1},
 			\end{equation} 
 			where $W$ are the diagonalizing matrices, $\lambda_r$ represents the real $r$-eigenvalue of $Y^{\alpha}\tau_{\alpha}$ and $B_r$ are matrices such that,
 			\begin{align*}
 			B_rB_l=\delta_{rl}B_{l},\qquad \sum_{r=1}^{m}B_r=\mathbb{I}_{m\times m}.
 			\end{align*}  
 			Then, the non-abelian unitary operator $U^{\tau}(p)$  can be rewritten in terms of the abelian unitary operator  $U(p)$  		on the domain $\mathcal{D}_{ {Q}}$ as follows,		
 			\begin{align}
 			U^{\tau}(p) =  
 			\sum_{r=1}^{m}	U( \lambda_r p)\otimes W  \, B_r W^{-1},\qquad\forall p\in\mathbb{R}^{d}.
 			\end{align} 
 		\end{lemma}
 		
 		\begin{proof}
 			From the definition of $Y^{\alpha}\tau_{\alpha}$ we know that it represents a hermitian matrix and that the eigenvalues different from zero are real.   By diagonalizing the matrix   $Y^{\alpha}\tau_{\alpha}= W D W^{-1}$ we  obtain a diagonal matrix  $D$ which we write as the sum of   real eigenvalues in the following manner,
 			\begin{equation}
 			D=\sum_{r=1}^{m} \lambda_r \, B_r,
 			\end{equation} 
 			where  $B_r$ are $m\times m$ matrices that have a one in the $r$-diagonal, i.e. in the $r$-th row and $r$-th column  and zero elsewhere. The properties of $B_r$ given in the lemma are easily verified. Next, on the domain $\mathcal{D}_{ {Q}}$ the Taylor expansion of  $U^{\tau}(p)$ is well defined and is given by, 
 			\begin{align*}
 			U^{\tau}(p)= \mathrm{exp}(i p^{\mu} Z_{{\mu}}( {X})\otimes Y^{\alpha}\tau_{\alpha}) 
 			&=\sum_{k=0}^{\infty}\frac{i^k}{k!} (p^{\mu} Z_{{\mu}})^k\otimes (W \sum_{r=1}^{m} \lambda_r \, B_r W^{-1})^k
 			\\
 			&=\sum_{r=1}^{m}\sum_{k=0}^{\infty} \frac{i^k}{k!} (\lambda_rp^{\mu} Z_{{\mu}})^k\otimes W  B_r W^{-1}  
 			\\&= \sum_{r=1}^{m} U(\lambda_r p)\otimes W  B_r W^{-1},
 			\end{align*} 
 			where in the last lines we used the unitarity of $W$ and  the power properties of $B_r$.
 		\end{proof}

 		In the next step we define   
 		$U^{\tau}(p)$ as  generator of the automorphisms $\alpha^{\tau}$.
 		
 		\begin{definition}\label{wcna}
 			Let $\Theta$ be a real skew-symmetric matrix on $\mathbb{R}^d$, let $A\in \mathcal{C}^{\infty} $ and let $E$ be
 			the spectral resolution of the unitary operator $U$. 
 			Then the corresponding non-abelian warped convolutions  of an operator $A\in\mathcal{C}^\infty$ with respect to $(\alpha^{\tau},\Theta)$ are the operators  $A^{{\Theta_{\tau}}}$ and  $_{{\Theta_{\tau}}}A$
 			of $A$  defined on the dense and stable domain $\mathcal{D}_{ {Q}}$ according to	\begin{align}
 			& A^{{\Theta_{\tau}}} :=\int dE(x)\alpha^{\tau}_{
 				\Theta x} (A\otimes \mathbb{I}_{m\times m}), \label{WC}
 			\\& _{{\Theta_{\tau}}}A :=\int \alpha^{\tau}_{
 				\Theta x} (A\otimes \mathbb{I}_{m\times m})  dE(x)\nonumber
 			\end{align}
 			where $\alpha^{\tau}$ denotes the adjoint action of $U^{\tau}$ given by $\alpha^{\tau}_{k}(A)=U^{\tau}(k)A U^{\tau}(k)^{-1}$.
 		\end{definition}
 		%% Could the difference of this with BLS be that we use the spectral measure $E$ of $U$, and not that of $U^\tau$?
 		
 		\begin{remark}
 			In literature concerning non-abelian fields, the $m\times m$-unit matrix often do not appear explicitly in formulas. Hence to ease readability we use this abuse of notation and do not write the unit matrices, unless necessarily in order to avoid confusion.
 		\end{remark} 
 	Concerning the commutative limit we have due to the former definition two possible limits. We can either take   $\Theta$ or $Y$  in the limit to zero. Another important distinction to the abelian case of \cite{BS} is the use of a non-abelian  adjoint action $U^{\tau}$ for   deformation while the spectral resolution is of the   unitary operator $U$. Hence  lemmas and theorems of the original work do not hold and have to be proven for the non-abelian adjoint action. However, it turns out that many proofs can be done analogously by using Lemma \ref{eis}.   Note, that the well-definedness of the integrals is given in the case of the deformed operator  $A\in \mathcal{C}^{\infty}$ since the non-abelian case is in a sense a $m\times m$ matrix of the original deformation. Hence for each component of the matrix we can use the same results of \cite{BLS} to argue the well-definedness. Nevertheless, we deform in this work unbounded operators hence in order to prove the validness of the deformation formula we need the non-abelian warped convolutions in terms of strong limits.

 	\begin{lemma}\label{eis2}
 		Let $\Theta$ be a real skew-symmetric matrix on $\mathbb{R}^d$, let $A\in \mathcal{C}^{\infty}$ and let $E$ be
 		the spectral resolution of the unitary operator $U$.  Then, the corresponding non-abelian warped convolution $A^{{\Theta_{\tau}}}$ and  $_{{\Theta_{\tau}}}A$
 		of $A$ are defined on the domain $\mathcal{D}_{ {Q}}$ in terms of strong limits as follows.
 		\begin{align} 
 		\int dE(x)\alpha^{\tau}_{
 			\Theta x} (A)     \Psi&=  (2\pi)^{-d}
 		\lim_{\epsilon\rightarrow 0}
 		\iint  d x\, d y \,\chi(\epsilon x,\epsilon y )\,e^{-ixy}\, U(y)\, \alpha^{\tau}_{\Theta x}(A)\Psi,\nonumber \\
 		\int \alpha^{\tau}_{
 			\Theta x} (A)  dE(x)   \Psi&=  (2\pi)^{-d}
 		\lim_{\epsilon\rightarrow 0}
 		\iint d x\, d y \,\chi(\epsilon x,\epsilon y )\,e^{-ixy}\, \alpha^{\tau}_{\Theta x}(A) U(y)\,\Psi
 		\end{align}
 		where $\alpha^{\tau}$ denotes the adjoint action of $U^{\tau}$ given by $\alpha^{\tau}_{k}(A )=U^{\tau}(k) \,A \, U^{\tau}(k)^{-1}$ and  $\chi \in\mathscr{S}(\mathbb{R}^d\times\mathbb{R}^d)$ with $\chi(0,0)=1$.
 	\end{lemma}
 	\begin{proof}
 		The proof can be done in an analogous manner to the proof in \cite{BLS}, where one uses first the fact that
 		\begin{equation*} 
 		U(y)=\int e^{ixy}dE(x),
 		\end{equation*}
 		hence one can deduce from this the following 
 		\begin{equation*} 
 		\int dE(x)=
 		(2\pi)^{-d}
 		\lim_{\epsilon\rightarrow 0}
 		\iint d x\, d y \,\chi(\epsilon x,\epsilon y )\,e^{-ixy}\, U(y).
 		\end{equation*}
 		Moreover since the non-abelian unitary operator can be rewritten in terms of the abelian one (see Lemma \ref{eis}) we have the following for     corresponding non-abelian warped convolutions $A^{{\Theta_{\tau}}}$    of $A$  
 		\begin{align*} 
 	A^{{\Theta_{\tau}}}=	\int dE(x)\alpha^{\tau}_{
 			\Theta x} (A)     \Psi
 		&=\sum_{r,l=1}^{m}\int dE(x)  U(
 			\lambda_r \Theta x)\, A \, U(
 			-\lambda_l \Theta x) \otimes W  B_rB_l W^{-1}\\&=
 	\sum_{r=1}^{m}	\int dE(x)  \alpha_{
 			\lambda_r \Theta x} (A)  \otimes W  B_r W^{-1} \\&= 	\sum_{r=1}^{m}	A_{{	\lambda_r \Theta }}\otimes W  B_r W^{-1},
 		 		\end{align*}
 		 	where $A_{{	\lambda_r \Theta }}$ is the abelian warped convolutions (see Definition \ref{defwca}) with the deformation matrix $\lambda_r \Theta$. For   operator $_{{\Theta_{\tau}}}A$ of $A$  we have analogously,
 		 		\begin{align*} 
 		 		_{{\Theta_{\tau}}}A=	\int \alpha^{\tau}_{
 		 			\Theta x} (A)   dE(x)  \Psi
 		 		&=\int  \sum_{r,l=1}^{m} U(
 		 		\lambda_r \Theta x)\, A \, U(
 		 		-\lambda_l \Theta x)  dE(x) \otimes W  B_rB_l W^{-1} \\&=
 		 		\sum_{r=1}^{m}	\int  \alpha_{
 		 			\lambda_r \Theta x} (A) dE(x)  \otimes W  B_r W^{-1} \\&= 	\sum_{r=1}^{m}	\,_{\lambda_r\Theta}A\otimes W  B_r W^{-1},
 		 		\end{align*}
 		where $_{{	\lambda_r \Theta }}A$ is the other possible definition of the  abelian warped convolutions (see Definition \ref{defwca}) with the deformation matrix $\lambda_r \Theta$. Therefore, the non-abelian deformed operator can be expressed as the abelian deformed one which is the left side of a tensor product. 
 		Hence, all the considerations concerning convergence and the strong limits of \cite{BLS} apply in this case as well, since the right hand side of the tensor product produced merely multiplication by matrices.   	
 		As already said the proof that the former formulas are well-defined for  $A\in \mathcal{C}^{\infty}$ is   analogous   to the proof in \cite[Section 2.2]{BLS}.
 	\end{proof}
 From the last proof an essential detail of this deformation emerged. The non-abelian deformation can be expressed as a tensor product, where the abelian case is on the left of this product, i.e.
 \begin{equation}\label{eqnaa}
 	A^{{\Theta_{\tau}}}=\sum_{r=1}^{m}	A_{{	\lambda_r \Theta }}\otimes W  B_r W^{-1}, \qquad \qquad 	_{{\Theta_{\tau}}}A=	\sum_{r=1}^{m}	\,_{\lambda_r\Theta}A\otimes W  B_r W^{-1}
 \end{equation}
  
 	In the next lemma we prove that the two deformations $A^{{\Theta_{\tau}}}$ and  $_{{\Theta_{\tau}}}A$ are equivalent. Moreover, we proof important statements concerning self-adjointness and commutativity which are essential in regards to physical (QM and QFT) aspects. 	
 	\begin{lemma}\label{wcl2}
 		Let $\Theta$ be a real skew-symmetric matrix on $\mathbb{R}^d$, let $A, F \in \mathcal{C}^{\infty}$. Then, on the dense domain $\mathcal{D}_{ {Q}}$ the following holds,\newline 
 		\begin{enumerate} \renewcommand{\labelenumi}{(\roman{enumi})}
 			\item   $\int \alpha^{\tau}_{\Theta x}(A)dE(x)=\int dE(x)\alpha^{\tau}_{\Theta x}(A)$
 			\item $\left(\int \alpha^{\tau}_{\Theta x}(A)dE(x)\right)^{*}\subset\int \alpha^{\tau}_{\Theta
 				x}(A^{*})dE(x)$
 			\item  Let $A, F \in \mathcal{C}^{\infty}$ be operators such that $[\alpha^{\tau}_{\Theta
 				x}(A),\alpha^{\tau}_{-\Theta y}(F)]=0$ for all $x,y \in sp \,U$. Then the non-abelian deformed operators commute, i.e. 
 			\begin{equation*}
 			[A^{{\Theta_{\tau}}}, F^{-{\Theta_{\tau}}}]=0.
 			\end{equation*}
 		\end{enumerate}
 	\end{lemma} 
 	\begin{proof} $(i)$ The proof is done by expressing the non-abelian deformation  in terms of strong limits using Lemma \ref{eis}, which gave a specific relation between the non-abelian and the abelian unitary operator. Hence, for $\Psi\in \mathcal{D}_{ {Q}}$ we have 
 		\begin{align*} 
 		&(2\pi)^d\int dE(x)\alpha^{\tau}_{
 			\Theta x} (A)\Psi=     	  
 		\lim_{\epsilon\rightarrow 0}
 		\iint  d x\, d y \,\chi(\epsilon x,\epsilon y )\,e^{-ixy}\, U(y)\, \alpha^{\tau}_{\Theta x}(A)\Psi, \\
 		&=      
 		\lim_{\epsilon\rightarrow 0}
 		\iint  d x\, d y \,\chi(\epsilon x,\epsilon y )\,e^{-ixy}\, U(y)\, U^{\tau}(\Theta x)(A)U^{\tau}(-\Theta x)\Psi, \\
 			&=     	  
 			\lim_{\epsilon\rightarrow 0} \sum_{r,l=1}^{m}\left(
 			\iint  d x\, d y \,\chi(\epsilon x,\epsilon y )\,e^{-ixy}  U(y) 	 U(\lambda_r \Theta x)(A)U(-\lambda_l \Theta x)\otimes W   {B_rB_l}  W^{-1}\right)\Psi, \\ 
 			&=     	  
 			\lim_{\epsilon\rightarrow 0}\sum_{r =1}^{m} \left(
 			\iint  d x\, d y \,\chi(\epsilon x,\epsilon y )\,e^{-ixy}\, U(y) 	 U(\lambda_r \Theta x)(A)U(-\lambda_r \Theta x)\otimes W  B_r  W^{-1}\right)\Psi, \\		&=     	  
 			\lim_{\epsilon\rightarrow 0}\sum_{r =1}^{m}\left (
 			\iint  d x\, d y \,\chi_1(\epsilon x,\epsilon y )\,e^{-ixy}\, 	 U(\lambda_r \Theta x)(A)U(-\lambda_r \Theta x+y)\otimes W  B_r  W^{-1}\right)\Psi,
 		\end{align*} 
 		where in the last lines we restricted the integration  to the submanifold $(\ker \Theta)^{\perp} \times (\ker \Theta)^{\perp}$, since
 		the remaining integrals merely produce factors of $2\pi$ and we performed the variable substitution $x \rightarrow x + (\lambda_r\Theta)^{-1} y$ 	and took into account that $\Theta^{-1}$ is skew symmetric. Furthermore, we defined a  cutoff function $\chi_1(\epsilon x,\epsilon y ):=\chi(\epsilon (x+  (\lambda_r\Theta)^{-1} y ),\epsilon y )$ and hence in the limit $\epsilon\rightarrow 0$ one obtains
 			\begin{align*} 
 		  	  &\lim_{\epsilon\rightarrow 0}\sum_{r =1}^{m}\left (
 		  	  \iint  d x\, d y \,\chi_1(\epsilon x,\epsilon y )\,e^{-ixy}\, 	 U(\lambda_r \Theta x)(A)U(-\lambda_r \Theta x+y)\otimes W  B_r  W^{-1}\right)\Psi\\&=
 		  	  \lim_{\epsilon\rightarrow 0}\sum_{r,l }^{m}\left (
 		  	  \iint  d x\, d y \,\chi_1(\epsilon x,\epsilon y )\,e^{-ixy}\, 	 U(\lambda_r \Theta x)(A)U(-\lambda_l \Theta x )U(y)\otimes { W  \delta_{rl}B_rW^{-1}}   \right)\Psi\\& = 
 		  	  \lim_{\epsilon\rightarrow 0} \left (
 		  	  \iint  d x\, d y \,\chi_1(\epsilon x,\epsilon y )\,e^{-ixy}\, 	 U^{\tau}(  \Theta x)(A)U^{\tau}(-\Theta x)U(y) \right)\Psi\\&
=	 (2\pi)^d\int \alpha^{\tau}_{
 			\Theta x} (A)dE(x)\Psi.		\end{align*} The importance to prove it here using strong limits is motivated by introducing to the reader how to work with   non-abelian deformations. However, since we know how the non-abelian deformation is related to the abelian one by Formula (\ref{eqnaa}), assertion $(i)$ follows from the original work by a one-liner, i.e.
 		\begin{align*}
 			A^{{\Theta_{\tau}}}&=\sum_{r=1}^{m}	A_{{	\lambda_r \Theta }}\otimes W  B_r W^{-1} =\sum_{r=1}^{m}\,_{{	\lambda_r \Theta }}A\otimes W  B_r W^{-1}=\, _{{\Theta_{\tau}}}A,
 		\end{align*}
 		where we used $A_{\Theta}=\,_{\Theta}A$ (see Equation \ref{eq1}) from the abelian case. \newline\newline
 	$(ii)$ The prove of the second assertion is done analogously to \cite[Lemma 2.2.]{BLS} by using the former assertion. 	   Since it is straightforward we skip the proof. \newline\newline
 	$(iii)$ This can be proven by rewriting the nonabelian operators in terms of the abelian deformed ones (see Lemma \ref{eis} and \ref{eis2}), i.e. 
 	 			\begin{align*} 
 	 			[A^{{\Theta_{\tau}}},F^{{\Theta_{\tau}}}]&=\sum_{l,r=1}^{m}	[A_{\lambda_r{\Theta}},F_{\lambda_l{\Theta}}]\otimes WB_rB_lW^{-1}\\&=
 	 			\sum_{ r=1}^{m}	[A_{\lambda_r{\Theta}},F_{\lambda_r{\Theta}}]\otimes WB_r W^{-1},
 			\end{align*} 
and hence from this point forward the proof of \cite[Proposition 2.10]{BLS} can be adapted to prove this assertion since the requirement $[\alpha^{\tau}_{\Theta
	x}(A),\alpha^{\tau}_{-\Theta y}(F)]=0$ can be translated in to the abelian case as follows,
	\begin{align*} 
[\alpha^{\tau}_{\Theta
	x}(A),\alpha^{\tau}_{-\Theta y}(F)]&=\sum_{l,r=1}^{m}[\alpha_{\lambda_r\Theta
	x}(A),\alpha_{-\lambda_l\Theta y}(F)]\otimes WB_rB_l W^{-1}\\&=
\sum_{ r=1}^{m}[\alpha_{\lambda_r\Theta
	x}(A),\alpha_{-\lambda_r\Theta y}(F)]\otimes WB_r  W^{-1}.
	\end{align*}

 	\end{proof}
 	In the original work \cite{BLS} the authors proved that warped convolutions supply isometric
 	representations of Rieffel's strict deformations of $C^{*}$-dynamical systems with
 	actions of $\mathbb{R}^d$, \cite{RI}. In the following lemma we introduce the deformed product, also known as the Rieffel product.
 	\newline
 	\begin{lemma}\label{l2.1}
 		Let $\Theta$ be a real skew-symmetric matrix on $\mathbb{R}^d$ and let $  {A},   {F} \in
 		\mathcal{C}^{\infty}$. Then 
 		\begin{equation*}
 		{A}_{\Theta}   {F}_{\Theta} \Phi= (A\times_{\Theta }F)_{\Theta}\Phi, \qquad
 		\Phi\in\mathcal{D}.
 		\end{equation*}
 		where $\times_{\Theta}$ is the \textbf{Rieffel product} on $\mathcal{C}^{\infty}$, defined by 
 		\begin{equation}\label{dp0}
 		(A\times_{\Theta}F )\Phi=(2\pi)^{-d}
 		\lim_{\epsilon\rightarrow 0}\iint dx\, dy \,\chi(\epsilon x,\epsilon y )\,e^{-ixy} \, \alpha_{\Theta x}(A)\alpha_{y}(F)\Phi.
 		\end{equation}
 	\end{lemma}$\,$\newline
 	Next we give a lemma stating how the non-abelian deformations supply the isometry.
	\begin{lemma}
		Let $\Theta$ be a real skew-symmetric matrix on $\mathbb{R}^d$ and let $A, F\in \mathcal{C}^{\infty}$.   Then, on the domain $\mathcal{D}_{ {Q}}$, we have
		\begin{equation*}
		{A} ^{{\Theta_{\tau}}}   {F}^{{\Theta_{\tau}}}= (A\times^{{\Theta_{\tau,\mathbb{I}}} }F)^{{\Theta_{\tau}}}.
		\end{equation*}
		where  the non-abelian Rieffel product is given 
		by a natural extension of the original definition for $\Phi\in\mathcal{D}_{ {Q}}$ as,  
			\begin{equation*} 
			(A\times^{{\Theta_{\tau,\mathbb{I}}}}F )\Phi=(2\pi)^{-d}
			\lim_{\epsilon\rightarrow 0}
			\iint dx\, dy \,\chi(\epsilon x,\epsilon y )\,e^{-ixy} \, \alpha^{\tau}_{\Theta x}(A\otimes \mathbb{I}_{m\times m})\alpha^{\mathbb{I}}_{y}(F\otimes \mathbb{I}_{m\times m})\Phi.
			\end{equation*}
		where the adjoint action $\alpha^{\mathbb{I}}_y$ is defined by the unitary operator $U(y)\otimes\mathbb{I}_{m\times m}$.

	\end{lemma}
	\begin{proof}For the following proof we make extensive use of Lemma \ref{eis}. The product of the non-abelian deformed operators $A,F\in \mathcal{C}^{\infty}$  reads on vectors of the domain $\mathcal{D}_{ {Q}}$,
		\begin{align*} {A} ^{{\Theta_{\tau}}}   {F}^{{\Theta_{\tau}}} &= \sum_{l,r=1}^{m}
		A_{\lambda_r{\Theta}}F_{\lambda_l{\Theta}}\otimes W  B_r B_l W^{-1}  \\&=
		\sum_{r=1}^{m}
		A_{\lambda_r{\Theta}}F_{\lambda_r{\Theta}}\otimes W  B_r W^{-1} \\&=
		\sum_{r=1}^{m}	(A\times_{\lambda_r\Theta }F)_{\lambda_r\Theta}\otimes W  B_r W^{-1} ,
		\end{align*}
		where in the last lines we used properties of the matrices $B_r$ and the isometry between warped convolutions and the Rieffel product given in Lemma \ref{l2.1}. Next we turn our attention to the  non-abelian Rieffel product given on $\mathcal{D}_{ {Q}}$   by, 
			\begin{align*}  (A\times^{{\Theta_{\tau,\mathbb{I}}}}F)&= 
			(2\pi)^{-d}
			\lim_{\epsilon\rightarrow 0}
			\iint dx\, dy \,\chi(\epsilon x,\epsilon y )\,e^{-ixy} \, \alpha^{\tau}_{\Theta x}(A\otimes \mathbb{I}_{m\times m})\alpha^{\mathbb{I}}_{y}(F\otimes \mathbb{I}_{m\times m}) \\&=  
			(2\pi)^{-d}\sum_{ r=1}^{m}
			\lim_{\epsilon\rightarrow 0}
			\iint dx\, dy \,\chi(\epsilon x,\epsilon y )\,e^{-ixy} \, \alpha_{\lambda_r\Theta x}(A )\alpha_{  y}(F )\otimes W  B_r  W^{-1} 
			\\&=  	\sum_{r=1}^{m}	(A\times_{\lambda_r \Theta }F) \otimes W  B_r W^{-1}.
			\end{align*}
			By deforming the  non-abelian Rieffel product $\times^{{\Theta_{\tau,\mathbb{I}}} }$ of $A,F$ the assertion follows. 
	\end{proof}
 This lemma is important in the context of algebraic quantum field theory and will be investigated in future works.
 	\\\\
 	Next we turn our attention to the well-definedness of the non-abelian deformation in the case of unbounded operators. 
 	To prove that the non-abelian deformation   (\ref{wcna}) holds as well for an unbounded operator
 	$A$, defined on a dense domain $\mathcal{D}(A)\subset \mathcal{H}$ of some separable Hilbert space $\mathcal{H}$, let us
 	consider the deformed operator $A^{{\Theta_{\tau}}}$ as follows
 	\begin{align*}
 	\langle \Psi,A^{{\Theta_{\tau}}}\Phi\rangle&=
 	(2\pi)^{-d}
 	\lim_{\epsilon\rightarrow 0}
 	\iint  \, dx \,  dy \, e^{-ixy}  \, \chi(\epsilon x,\epsilon
 	y) {\langle \Psi, 
 		U(y)\alpha^{\tau}_{\Theta x}(A)\Phi\rangle} \nonumber \\&
 	=:
 	(2\pi)^{-d}
 	\lim_{\epsilon\rightarrow 0}
 	\iint  \, dx\,  dy \, e^{-ixy}  \, \chi(\epsilon x,\epsilon
 	y) \, b(x,y)
 	\end{align*}
 	for $\Psi, \Phi \in \mathcal{D}^{\infty}(A)\otimes \mathbb{C}^m:=\{
 	\varphi \in\mathcal{D}(A)| U(x)\varphi \in \mathcal{D}(A)$  is smooth in $\|\cdot\|_{\mathcal{H}}
 	\}\otimes \mathbb{C}^m$ and 	   $\chi \in\mathscr{S}(\mathbb{R}^d\times\mathbb{R}^d)$ with $\chi(0,0)=1$. Note that   the scalar product is   w.r.t. $\mathcal{H}
 	\otimes \mathbb{C}^m$.
 	The expression is well-defined if $b(x,y)$ is a symbol which is given in the following definition (\cite[Section 7.8, Definition 7.8.1]{H}).

 	\begin{definition}
 		Let $r$, $\rho$, $\delta$, be real numbers with $0<\rho\leq1$ and $0\leq \delta <1$. Then we
 		denote by $S^{r}_{\rho,\delta}(X\times \mathbb{R}^d)$, the set of all $b\in C^{\infty}(X\times
 		\mathbb{R}^d)$ such that for every compact set $K\subset X$ and all  $\gamma, \kappa$  the estimate 
 		\begin{equation*}
 		|\partial_{x}^{\gamma}\partial_{y}^{\kappa}b(x,y)|\leq
 		C_{\gamma,\kappa,K}(1+|x|)^{r-\rho|\gamma|+\delta|\kappa|},\qquad x\in K, \,\, y\in
 		\mathbb{R}^d,  
 		\end{equation*}
 		is valid for some constant $C_{\gamma,\kappa,K}$. The elements  $S^{r}_{\rho,\delta}$ are called
 		symbols of order $r$ and type $\rho,\delta$. Note that  $\gamma, \kappa$ are multi-indices and  $|\gamma|$, $|\kappa|$ are the corresponding sums of the index-components. 
 	\end{definition}

 	\begin{lemma}\label{lfhs}
 		Assume that the derivatives of the adjoint action of $A$ w.r.t. the unitary operator $U^{\tau}$ are polynomially bounded on vectors in $\mathcal{D}^{\infty}(A)\otimes \mathbb{C}^m$, i.e.
 		\begin{equation}\label{pb}
 		\|\partial_{x}^{\gamma}\alpha^{\tau}_{\Theta x}(A)\Phi\| \leq  C_{\gamma}(1+|x|)^{r-\rho|\gamma|},\qquad \forall \Phi \in
 		\mathcal{D}^{\infty}(A)\otimes \mathbb{C}^m.
 		\end{equation}
 		Then, $b(x,y)$ belongs to the symbol class $S^{r}_{\rho,0}$ for $\Psi, \Phi \in
 		\mathcal{D}^{\infty}(A)\otimes \mathbb{C}^m$ and therefore the non-abelian deformation  of the unbounded operator $A$ is given as a well-defined
 		oscillatory integral.
 	\end{lemma} 
 	\begin{proof}
 		For the derivatives of the  scalar product $b(x,y)$ the following estimates hold 
 		\begin{align*} \partial_{x}^{\gamma}\partial_{y}^{\kappa}
 		{\langle \Psi, 
 			U(y)\alpha^{\tau}_{\Theta x}(A)\Phi\rangle}&\leq \| (-iZ)^{\kappa} \Psi\| \| 
 		\partial_{x}^{\gamma}\alpha^{\tau}_{\Theta x}(A)\Phi\|\\&\leq\underbrace{
 			\| (-iZ)^{\kappa} \Psi\|  C_{\gamma}}_{ =:C_{\gamma,\kappa,K}}(1+|x|)^{r-\rho|\gamma|},
 		\end{align*}
 		where in the last lines we used Assumption (\ref{pb}). To prove the second part of the statement we first define $r(\rho):=r-\rho|\gamma|$ and use the former inequality, i.e. 
 		\begin{align*}
 		&(2\pi)^{-d}
 		\lim_{\epsilon\rightarrow 0}
 		\iint  \, dx\,  dy \, e^{-ixy}  \, \chi(\epsilon x,\epsilon
 		y) \,\partial_{x}^{\gamma}\partial_{y}^{\kappa} b(x,y)
 		\\ &\leq(2\pi)^{-d}
 		\lim_{\epsilon\rightarrow 0}
 		\iint  \, dx\,  dy \, e^{-ixy}  \, \chi(\epsilon x,\epsilon
 		y) \,C_{\gamma,\kappa,K}(1+|x|)^{r(\rho)}
 		\\&\leq  (2\pi)^{-d} C_{\gamma,\kappa,K}
 		\lim_{\varepsilon_1\rightarrow 0}  \left(
 		\int dx \lim_{\varepsilon_2\rightarrow 0} 
 		\left(\int dy  e^{-ixy}
 		\chi_2(\varepsilon_2 y)\right)\,\chi_1(\varepsilon_1  x)\,
 		\right)(1+|x|)^{r(\rho )}
 		\\
 		&=       (2\pi)^{-d/2}C_{\gamma,\kappa,K}
 		\lim_{\varepsilon_1\rightarrow 0}  \left(
 		\int dx \,
 		\delta( {x} )\,\chi_1(\varepsilon_1  x)   (1+|x|)^{r(\rho )} \right) =C_{\gamma,\kappa,K}.
 		\end{align*}
 		Note that the oscillatory integral does not depend on the chosen cut-off function. Hence, we can proceed as in
 		\cite{RI}, where the regulator is chosen as  $\chi (\varepsilon x,\varepsilon y)= \chi_2(\varepsilon_2 x
 		)\chi_1(\varepsilon_1 y)$ 
 		with $\chi  \in\mathscr{S}(\mathbb{R}^{d}\times\mathbb{R}^{d})$ and $\chi_{l}(0 )=1$, $l=1$, $2$,
 		and we obtained the delta
 		distribution $\delta(x)$ in the limit $\varepsilon_2 \rightarrow
 		0$, \cite[Section 7.8, Equation 7.8.5]{H}.   For another proof of the second statement we refer the reader to \cite{LW} and \cite[Theorem 1]{AA}. 
 	\end{proof}
 	Next we concentrate on the self-adjointness of an operator $A$ that is deformed with non-abelian warped convolutions where the  undeformed version is self-adjoint. 
 	\begin{theorem}\label{lfhs1}
 		The operator $A^{{\Theta_{\tau}}}-A\otimes\mathbb{I}_{m\times m}$ is a symmetric operator on the dense domain $\mathcal{D}^{\infty}(A)\otimes \mathbb{C}^m$. Furthermore, let $A$ be an unbounded self-adjoint operator that fulfills the assumptions of polynomial boundedness w.r.t. the non-abelian adjoint action. Then $A^{{\Theta_{\tau}}}$ is essentially self-adjoint on  $\mathcal{D}^{\infty}(A)\otimes \mathbb{C}^m$.
 	\end{theorem} 
 	\begin{proof}
 		Since the undeformed operator $A\otimes\mathbb{I}_{m\times m}$ is a self-adjoint operator for all vectors of $\mathcal{D} (A)\otimes \mathbb{C}^m$ and $\mathcal{D}^{\infty}(A)\otimes \mathbb{C}^m$ is a dense subset of the domain, symmetry follows. For the first part of the operator, i.e. $A^{{\Theta_{\tau}}}$ symmetry follows easily from Lemma \ref{wcl2}. The proof of self-adjointess can be done  along the same lines as in \cite[Theorem 5.1]{MUc5}.
 	\end{proof} 
 	\section{Application of Deformation in  Quantum Mechanics}
 	
 	In the so called \textbf{Schr\"odinger
 		representation}, \cite{RS1, BEH} the pair of operators
 	$(P_{i},X_{j})$, satisfying the \textbf{canonical commutation relations}  (CCR)
 	\begin{equation}\label{ccr}
 	[X_i,P_{k}]=-i\eta_{ik},
 	\end{equation}
 	are represented as essentially self-adjoint operators on the dense domain
 	$\mathscr{S}(\mathbb{R}^n)$. Here $X_{i}$ and $P_{k}$ are the closures of
 	$x_{i}$ and multiplication by $i {\partial}/{\partial x^k}$ on
 	$\mathscr{S}(\mathbb{R}^n)$ respectively. The dense domain representing  $\mathcal{D}^{\infty}(A)\otimes \mathbb{C}^m$ for   $A$ being the coordinate and/or momentum operator is given by $D_{\mathcal{E}}:=\mathcal{E}\otimes \mathbb{C}^m$, where $\mathcal{E}\subseteq \mathscr{S}(\mathbb{R}^n)$ (see \cite[Lemma 3.1]{Muc1}). 
 	\subsection{Non-Abelian Gauge Fields}
 	In this Section  we  generate nonabelian gauge fields by applying the deformation procedure on the momentum operator  in quantum mechanics.  The infinitesimal generator of deformations is an unbounded real-vector-valued function of the coordinate operator. 	\begin{proposition}
 		The quantum mechanical momentum operator deformed with the non-abelian  warped convolutions, i.e. $\vec{P}^{\Theta_{\tau}}$, by using   as generator   $Q_{k}=Z(\mathbf{X})_{k}\otimes Y_{\alpha}  \tau^{\alpha}$, is well-defined on  the Hilbert-space $D_{\mathcal{E}}$  and is given explicitly by the following equation 
 		\begin{equation} 
 		{P}_{i}^{\Theta_{\tau}}= {P}_{i}   -q {A}_{i,\alpha}  \tau^{\alpha},\end{equation}
 		where the gauge field $\vec{A}_{\alpha}$ is given as

 		\begin{equation} 
 		-q\vec{A}_{\alpha} :=(\Theta Z(\mathbf{X}))_k \vec{\partial} Z(\mathbf{X})^{k}\otimes Y_{\alpha}. \end{equation}
 	\end{proposition}
 	
 	\begin{proof}
 		First, we  apply the new definition of deformation on the momentum operator and since we have shown the well-definedness through the equivalence of the spectral measure and the deformation in terms of strong limits we have,
 		\begin{align*} 
 		{P}^{\Theta_{\tau}}_{i }\Psi&=\int \alpha^{\tau}_{\Theta x}({P}_{i} \otimes \mathbb{I}_{m\times m} ) dE (x)  \Psi\\&=
 		\int\left( U^{\tau}(\Theta x) ({P}_{i} \otimes \mathbb{I}_{m\times m} )  U^{\tau}(\Theta x)^{-1}\right)dE (x)\Psi\\&=
 		\int \left(\vec{P} +i ( \Theta x)_k [{Q}(\mathbf{X})^{k}, ({P}_{i} \otimes \mathbb{I}_{m\times m} )  ]+  \,\cdots
 		\right)
 		dE (x)\Psi
 		,
 		\end{align*} 
 		where in the last line we used the Backer-Campbell-Hausdorff formula. Let us take a closer look at the first commutator 
 		\begin{align*} 
 		[{Q}(\mathbf{X})^{k}, {P}_{i} \otimes \mathbb{I}_{m\times m} ]  =-i
 		{\partial}_{i} Z(\mathbf{X})^{k}\otimes Y_{\,\,\alpha} \tau^{\alpha},
 		\end{align*} 
 		which is calculated by using the CCR. Moreover, since the commutator gave an operator only depending on the coordinate operator we deduce that all other commutators vanish, i.e.
 		\begin{align*} 
 		[{Q}(\mathbf{X})^{j},[{Q}(\mathbf{X})^{k}, {P}_{i} \otimes \mathbb{I}_{m\times m} ]]=0.
 		\end{align*} 
 		Hence, by collecting all terms we have 
 		\begin{align*} 
 		{P}^{\Theta_{\tau}}_{i }\Psi&=
 		\int 	dE (x) \left({P}_{i}  +i  (\Theta x)_k [{Q}(\mathbf{X})^{k},{P}_{i} \otimes \mathbb{I}_{m\times m}    ]\right)
 	 \Psi\\&=
 		\int		dE (x)  \left({P}_{i} \otimes \mathbb{I}_{m\times m}  + (\Theta x)_k   {\partial}_{i} Z(\mathbf{X})^{k}\otimes Y_{\alpha} \tau^{\alpha}\right)
 \Psi\\&= \left({P}_{i} \otimes \mathbb{I}_{m\times m}  +\underbrace{( \Theta Z)_k  {\partial}_{i}  Z(\mathbf{X})^{k}\otimes Y_{\alpha}}_{-q\vec{A}_{\alpha}} \tau^{\alpha}\right)\Psi\\&=  \left({P}_{i} \otimes \mathbb{I}_{m\times m}   -q {A}_{i,\alpha}  \tau^{\alpha}\right)\Psi
 		\end{align*}  
 		The deformation formula is well-defined even though the operator is unbounded. This is due to the vanishing of all orders, except the first order, of the Baker-Campbell Hausdorff formula. Hence, the first-order term is polynomially bounded in $x$. Let us examine the expression of polynomial boundedness more careful, i.e.
 		\begin{align*}
 		\|\partial_{x}^{\gamma}\alpha^{\tau}_{\Theta x}({P}_{i})\Psi\| &=
 		\|\partial_{x}^{\gamma} 	  \left({P}_{i} \otimes \mathbb{I}_{m\times m}  + (\Theta x)_k    {\partial}_{i} Z(\mathbf{X})^{k}\otimes Y_{\alpha} \tau^{\alpha}\right)\Psi\|
 		\end{align*}
 		For the multi-index $\gamma$ equal to zero we have 
 		\begin{align*}
 		\|   \left({P}_{i}   + (\Theta x)_k   {\partial}_{i} Z(\mathbf{X})^{k}\otimes Y_{\alpha} \tau^{\alpha}\right)\Psi\|&\leq \| {P}_{i} \Psi\|+ \|\left((\Theta x)_k   {\partial}_{i} Z(\mathbf{X})^{k}\otimes Y_{\alpha} \tau^{\alpha}\right)\Psi\|\\&\leq
 		\| {P}_{i} \Psi\|+|\vec{x}|\,\|\left((\Theta_{kl}e^{l}  {\partial}_{i} Z(\mathbf{X})^{k}\otimes Y_{\alpha} \tau^{\alpha}\right)\Psi\|\\&\leq C_{0}(1+|\vec{x}|),
 		\end{align*}
 		where in the last lines we used Cauchy-Schwarz and the fact that $\Psi\in D_{\mathcal{E}}$ such that the norms of the respective operators are bounded and hence a constant $C_{0}$ can be found such that the inequality holds. 	For the multi-index $\gamma$ equal to one  the inequality is simply bounded by a constant. Due to the first-order boundedness in $x$ it follows from Lemma \ref{lfhs} that the deformation is well-defined (see also \cite{LW}, \cite[Theorem 1]{AA} and \cite{Muc1}).   
 	\end{proof} 	Next we investigate the case  of  a free Hamiltonian. For a free quantum-mechanical particle the energy is described by the  operator $H_{0}$ which is given as
 	\begin{equation}\label{fh}
 	H_{0}= -\frac{P_jP^j}{2m}.
 	\end{equation}   
 	In the case of non-abelian deformation there are two expressions that are worth investigating, i.e. \newline
 	\begin{itemize} 
 		\item    $H^{\Theta_{\tau}}_{0}:=-(P_{i}P^{i})^{\Theta_{\tau}}$\newline
 		\item   $H^{\Theta_{\tau}}_{1}:=-P^{\Theta_{\tau}}_{i}P^{i}_{\Theta_{\tau}}=-(P_{i}\times^{{\Theta_{\tau,\mathbb{I}}}}P^{i})^{\Theta_{\tau}}$   .
 	\end{itemize} $\,$\newline
 	The first case is the Hamiltonian plugged in to the Formula (\ref{wcna}) of non-abelian deformation. However another possibility of defining the deformed Hamiltonian is   by taking the scalar product of the non-abelian momentum operators.   
 	\begin{proposition}
 		The free Hamiltonian deformed with the non-abelian deformations,   by using  as generator  $Q_{k}=Z(\mathbf{X})_{k}\otimes Y_{\alpha}  \tau^{\alpha}$,  is well-defined  on the Hilbert-space  $D_{\mathcal{E}}$  and  regardless of the possible definitions unique. The explicit expression is given by 
 		\begin{equation} 
 		H^{\Theta_{\tau}}_{0}= - ({P}_{i}   -q {A}_{i,\alpha}  \tau^{\alpha}) ({P}^{i}   -q {A}^{i}_{\beta}  \tau^{\beta}) ,\end{equation}
 		where the gauge field $\vec{A}_{\alpha}$ is	
 		\begin{equation} 
 		-q\vec{A}_{\alpha} :=(\Theta Z(\mathbf{X}))_k \vec{\partial} Z(\mathbf{X})^{k}\otimes Y_{\alpha}. \end{equation}
 	\end{proposition}
 	
 	\begin{proof}
 		For the proof we shall use the spectral definition of the deformation. This can be done since the Hamiltonian is quadratically bounded w.r.t. the unitary operator $U^{\tau}$. This fact is easily seen since the adjoint action     is given by
 		\begin{align*} 
 		-\alpha^{\tau}_{\Theta x}(H_{0}  )&=U^{\tau}(\Theta x)(P_iP^i\otimes \mathbb{I}_{m\times m} ) U^{\tau}(\Theta x)^{-1}
 		\\&=U^{\tau}(\Theta x)(P_i\otimes \mathbb{I}_{m\times m} ) (P^i\otimes \mathbb{I}_{m\times m} )U^{\tau}(\Theta x)^{-1}
 		\\&=  U^{\tau}(\Theta x) (P_i\otimes \mathbb{I}_{m\times m} )U^{\tau}(\Theta x)^{-1}\,U^{\tau}(\Theta x)(P^i\otimes \mathbb{I}_{m\times m} ) U^{\tau}(\Theta x)^{-1},
 		\end{align*} 	
 		where the momentum operator is polynomially bounded to first order as was shown in the former proposition. Hence quadratic polynomial boundedness follows and by using Proposition \ref{lfhs}, one concludes that the deformation is well-defined. Therefore the integral   or  spectral representation are equivalent. Next we turn to the explicit result of  the deformed  free Hamiltonian in the non-abelian context. This is done by applying the deformation formula,
 		\begin{align*} 
 	-	H^{\Theta_{\tau}}_{0}\Psi:&=(P_{i}P^{i})^{\Theta_{\tau}}\Psi=\int dE(x)\,  \alpha^{\tau}_{\Theta x}(H_{0}  )  \Psi\\&=
 		\int dE(x)\left( U^{\tau}(\Theta x) \,H_{0}  \, U^{\tau}(\Theta x)^{-1}\right)\Psi\\&= 
 		\int dE(x)\left( U^{\tau}(\Theta x) \,P_i\,U^{\tau}(\Theta x)^{-1}\,U^{\tau}(\Theta x)\, P^i  \, U^{\tau}(\Theta x)^{-1}\right)\Psi\\&=
 		\int dE(x)\left({P}_{i}  +i  (\Theta x)_k [{Q}(\mathbf{X})^{k},{P}_{i}  ]\right) \left({P}^{i}  +i  (\Theta x)_r [{Q}(\mathbf{X})^{r},{P}^{i}  ]\right) \Psi\\&
 		= \left({P}_{i} + ( \Theta Z)_k  {\partial}_{i}  Z(\mathbf{X})^{k}\otimes Y_{\alpha}    \tau^{\alpha}\right) \left({P}^{i} + ( \Theta Z)_r  {\partial}^{i}  Z(\mathbf{X})^{r}\otimes Y_{\alpha}  \tau^{\alpha}\right) \Psi\\&
 		= \left({P}_{i}	-q {A}_{i,\alpha}    \tau^{\alpha}\right) \left({P}^{i}   	-q{A}^{i}_{\alpha}  \tau^{\alpha}\right) \Psi
 		.
 		\end{align*} 	
 		Next we turn our attention to defining the non-abelian deformed Hamiltonian as the scalar product of the non-abelian deformed momentum operators
 		\begin{align*} 
 	-	H^{\Theta_{\tau}}_{1}\Psi:=P^{\Theta_{\tau}}_{i}P^{i}_{\Theta_{\tau}}\Psi&= 
 		\left({P}_{i}  -q {A}_{i,\alpha}  \tau^{\alpha}\right)\left({P}^{i}  -q {A}^{i}_{\alpha}  \tau^{\alpha}\right)\Psi=
 		-H^{\Theta_{\tau}}_{0}\Psi,
 		\end{align*} 	
 		where we used the former proposition in order to obtain the result.   We conclude that all possible deformations in this case are equivalent. \\\\
 		The proof that the deformed Hamilton operator is well-defined   can be done as the former result concerning the momentum operators, i.e. due to the fact that the Backer-Campbell-Hausdorff formula vanishes to third order polynomial boundedness to second-order follows. Hence, as before by using Cauchy-Schwarz and domain arguments inequalities to second order in $x$ hold and we can use Lemma \ref{lfhs}.  
 	\end{proof} 
 	
 	In \cite{Muc1} it was proven that by identifying the deformation matrix with certain physical constants the deformed momentum represents the minimally substituted momentum operator. In this context, the commutator of two deformed momentum operators gave the spatial part of the field strength tensor, which either represents the magnetic or gravito-magnetic field. In the following we shall perform the same calculation with the non-abelian deformed momentum operators in order to prove that the outcome is a spatial non-abelian field strength tensor.

 	\begin{lemma}
 		The commutator of non-abelian deformed momentum operators   is given by the following term
 		\begin{align*}
 		[P^{\Theta_{\tau}}_{i},P^{\Theta_{\tau}}_{j}]
 		=-iq(\partial_{i}{A}_{j,\alpha}-  \partial_{j}{A}_{i,\alpha})\tau^{\alpha},
 		\end{align*}
 		which resembles the spatial field strength tensor without the quadratic term in the gauge-fields.
 	\end{lemma}
 	\begin{proof}
 		The explicit calculation is given by 
 		\begin{align*}
 		[P^{\Theta_{\tau}}_{i},P^{\Theta_{\tau}}_{j}]&=
 		[{P}_{i}  -q {A}_{i,\alpha}  \tau^{\alpha},{P}_{j}  -q {A}_{j,\beta}  \tau^{\beta}]\\&=
 		-q[{P}_{i},  {A}_{j,\beta}  \tau^{\beta}]-q [{A}_{i,\alpha}  \tau^{\alpha},{P}_{j}  ]+q^2
 		[{A}_{i,\alpha}  \tau^{\alpha},{A}_{j,\beta}  \tau^{\beta}]\\&
 		=-iq(\partial_{i}{A}_{j,\alpha}-  \partial_{j}{A}_{i,\alpha})\tau^{\alpha}+ q^2  \left({A}_{i}   {A}_{j}\otimes Y_{\alpha}\tau^{\alpha}Y_{\beta} \tau^{\beta}-{A}_{j}   {A}_{i}\otimes Y_{\alpha}\tau^{\alpha}Y_{\beta} \tau^{\beta}\right) \\&
 		=-iq(\partial_{i}{A}_{j,\alpha}-  \partial_{j}{A}_{i,\alpha})\tau^{\alpha}+ q^2 
 		\underbrace{\left([{A}_{i},{A}_{j}]\otimes Y_{\alpha}\tau^{\alpha}   Y_{\beta} \tau^{\beta} \right)}_{=0}
 		\end{align*}
 		where the gauge field $-qA_i= (\Theta Z(\mathbf{X}))_k  {\partial}_{i} Z(\mathbf{X})^{k}$.
 	\end{proof}
 	
 	Due to the vanishing of  quadratic terms of gauge fields in the spatial part of the field strength tensor, we call the deformation almost non-abelian.  However, note that the quadratic terms in the Hamiltonian do not vanish, i.e. 
 	\begin{align*}
 	q^2{A}_{i,\alpha}  \tau^{\alpha} {A}^{i}_{\beta}  \tau^{\beta} = q^2{A}_{i}{A}^{i}\otimes Y_{\alpha}  \tau^{\alpha} Y_{\beta} \tau^{\beta}.
 	\end{align*}

 	\subsection{Non-Abelian Moyal-Weyl}
 	In \cite{Muc1} we found a connection between deforming the momentum operator and the Moyal-Weyl space. This connection was given by the famous Landau quantization. We obtained the Landau quantization by deforming the momentum operator with the coordinate operator. To obtain the so called guiding center coordinates, that describe precisely the position of the particle we deformed the coordinate operator with the momentum operator. \\\\ Since, we obtained non-abelian gauge fields by deforming the momentum operator we reuse the idea found in QM and deform the coordinate operator. The resulting space-time  of this deformation is called the   non-abelian Moyal Weyl plane.  In order to obtain this new plane we have to deform the coordinate operator by using the non-abelian warped convolutions.  Afterwards the commutator of the two non-abelian deformed operators is calculated.
 	
 	\begin{proposition} \textbf{Non-Abelian Moyal-Weyl Plane}\newline 
 		Let the generator of deformations be given as the momentum operator i.e. $Q_{k}=P_{k}\otimes Y_{\alpha}  \tau^{\alpha}$. Then, the deformed coordinate operator $\vec{X}^{\Theta_{\tau}}$ is well-defined   on the Hilbert-space  $D_{\mathcal{E}}$ and given explicitly  by the following equation 
 		\begin{equation*} 
 		{X}^{i,\Theta_{\tau}}={X}^{i } -   (\Theta P)^{i} \otimes Y_{\alpha} \tau^{\alpha}.
 		\end{equation*}
 		Moreover, the non-abelian Moyal-Weyl plane is generated by the algebra of the non-abelian deformed coordinate operators  and is given as follows, 
 		\begin{equation}
 		[ {X}^{i,\Theta_{\tau}}, {X}^{j,\Theta_{\tau}}]=-2i\Theta^{ij}\otimes  Y_{\alpha} \tau^ {\alpha}.
 		\end{equation} 
 		
 	\end{proposition}
 	
 	\begin{proof}
 		First, we  apply the new definition of deformation on the momentum operator,
 		\begin{align*} 
 		{X}_{i}^{\Theta_{\tau}}\Psi&=\int \alpha^{\tau}_{\Theta p}({X}_{i }  )  dE(p) \Psi\\&=
 		\int\left( U^{\tau}(\Theta p) \,{X}_{i }  \, U^{\tau}(\Theta p)^{-1}\right)dE(p)\Psi\\&=
 		\int \left({X}_{i }  +i  (\Theta p)_k [{Q}(\mathbf{P})^{k} , {X}_{i } ]+  \,\cdots
 		\right)
 		dE(p)\Psi
 		,
 		\end{align*} 
 		where in the last line we used the Baker-Campbell-Hausdorff formula. Let us take a closer look at the first commutator 
 		\begin{align*}  [Q(\mathbf{P})^{k},  {X}_{j} ]   = [{P}^{k} ,  {X}_{j} ]  \otimes Y_{\alpha}\tau^{\alpha}  = i\delta_{j}^{k}\otimes Y_{\alpha}\tau^{\alpha},
 		\end{align*} 
 		where in the last line we used the CCR. Moreover, since the commutator gave an operator only depending on the coordinate operator we deduce that all other commutators vanish, i.e.
 		\begin{align*} 
 		[{Q}(\mathbf{P})_{j},[{Q}(\mathbf{P})_{k}, {X}^{i } ]]=0.
 		\end{align*} 
 		Hence, by collecting all terms we have 
 		\begin{align*} 
 		{X}_{i}^{\Theta_{\tau}}\Psi&=
 		\int \left( 
 		{X}_{i } +i  ( \Theta p)_k   [{Q}(\mathbf{P})^{k}, 
 		{X}_{i } ]\right)
 		dE(p) \Psi\\&=
 		\int \left(
 		{X}_{i } - (\Theta p)_i \otimes   Y^{\alpha}\tau_{\alpha} \right)
 		dE(p) \Psi\\&= \left( 
 		{X}_{i }-  (\Theta P)_{i}\otimes Y_{\alpha} \tau^ {\alpha}\right)\Psi .
 		\end{align*}  
 		The calculation of the non-abelian Moyal-Weyl plane is done in a straight forward manner, i.e.
 		\begin{align*}
 		[ {X}^{i,\Theta_{\tau}}, {X}^{j,\Theta_{\tau}}]
 		&=[{X}^{i }  -  \Theta^{i}_{\,\,k} P^k \otimes Y_{\alpha} \tau^ {\alpha},{X}^{j }  -   \Theta^{j}_{\,\,r} P^r \otimes  Y_{\beta} \tau^ {\beta}]\\&=- [\Theta^{i}_{\,\,k} P^k \otimes Y_{\alpha} \tau^ {\alpha},{X}^{j } ]+i\leftrightarrow j\\&=-\Theta^{i}_{\,\,k}[ P^k,{X}^{j }]\otimes   Y_{\alpha}  \tau^ {\alpha}+i\leftrightarrow j\\&=-i
 		\Theta^{ij} \otimes Y_{\alpha} \tau^ {\alpha}+i\leftrightarrow j\\&=-2i\Theta^{ij}  \otimes Y_{\alpha} \tau^ {\alpha}.
 		\end{align*}  
 		The proof that the deformation formula is well-defined for the unbounded deformed coordinate operator is equivalent to the proof done for the case of the momentum operator.  This is due to the vanishing of all orders, except the first order, of the Baker-Campbell Hausdorff formula. Therefore, the first-order term is polynomially bounded in $p$, i.e.
 		\begin{align*}
 		\|\partial_{p}^{\gamma}\alpha^{\tau}_{\Theta p}({X}_{i})\Psi\| &=
 		\|\partial_{p}^{\gamma} 	  \left({X}_{i} \otimes \mathbb{I}_{m\times m}  + (\Theta p)_i     \otimes Y_{\alpha} \tau^{\alpha}\right)\Psi\|
 		\end{align*}
 		For the multi-index $\gamma$ equal to zero we have 
 		\begin{align*}
 		\|   \left({X}_{i}   - (\Theta p)_i     \otimes Y_{\alpha} \tau^{\alpha}\right)\Psi\|&\leq \| {X}_{i} \Psi\|+ \|\left((\Theta p)_i\otimes Y_{\alpha} \tau^{\alpha}\right)\Psi\|\\&\leq
 		\| {X}_{i} \Psi\|+|\vec{p}|\,\|\left((\Theta_{il}e^{l}  \otimes Y_{\alpha} \tau^{\alpha}\right)\Psi\|\\&\leq \tilde{C}_{0}(1+|\vec{p}|),
 		\end{align*}
 		where as before we used Cauchy-Schwarz and the fact that $\Psi\in D_{\mathcal{E}}$ such that the norms of the respective operators are bounded and hence a finite constant $\tilde{C}_{0}$ can be found such that the inequality holds. 	 Due to the inequality it follows from Lemma \ref{lfhs} that the deformation is well-defined (see also \cite{LW}, \cite[Theorem 1]{AA} and \cite{Muc1}).   
 	\end{proof}

 	Another important point in the context of the non-abelian operators is the question of self-adjointness. Hence in order to have a $*$-algebra the non-abelian deformed coordinate operator has to be self-adjoint. 
 	\begin{lemma}
 		The non-abelian deformed operators ${X}^{i,\Theta_{\tau}}$ are self-adjoint   on the domain $ \mathscr{S}(\mathbb{R}^n)\otimes{\mathbb{C}^{m}}$.
 	\end{lemma}
 	\begin{proof}
 		Self-adjointness in the deformational context is owed to the polynomial boundedness of the expression and follows from Theorem \ref{lfhs1}.
 	\end{proof}

 	\section{QFT on Non-Abelian Moyal-Weyl}
 	In  \cite{GL1}, the authors gave a correspondence between free deformed fields and quantum fields living on the constant Moyal-Weyl space. The correspondence was implemented by a unitary operator mapping   the Fock space $\mathcal{H}$ to the tensor product space  $\mathcal{V}\otimes\mathcal{H}$, where the space $\mathcal{V}$  is defined as follows.     
 
 	\begin{definition}
 		 	Let $\mathcal{V}$ denote  the representation space of the *-algebra  generated by   self-adjoint operators $\hat{x}$ that fulfill  
 		 	\begin{equation}\label{xcr1}
 		 	[\hat{x}_{\mu},\hat{x}_{\nu}]= 2i\Theta_{\mu\nu},
 		 	\end{equation}
 		 	where $\Theta$ is the center of the algebra.
 	\end{definition}
 	 
The unitary operator responsible for the equivalence between the $*$-algebras of the deformed  fields  and the   $*$-algebra of   fields on  $\mathcal{V}$, \cite{GL1} is given by   $V_{\xi}= \bigoplus_{n=0}^{\infty}V^{(n)}_{\xi}:\mathcal{H}\rightarrow\mathcal{V}\otimes\mathcal{H}$,  with $\xi\in\mathcal{V}$ and $||\xi||_{\mathcal{V}}=1$,
 	\begin{equation}\label{nc1}
 	\left(V^{(n)}_{\xi}\Psi_n\right)\left(\mathbf{p}_1,\dots,\mathbf{p}_n\right)=\Psi_n\left(\mathbf{p}_1,\dots,\mathbf{p}_n\right)\cdot e^{i\sum\limits_{k=1}^{n}p_{k}\hat{x}}\xi,\qquad \Psi_n\in\mathcal{H}_n.
 	\end{equation}
 	Hence,   the following equations hold in a distributional sense 
 	\begin{align}\label{aot} a_{\otimes}(\mathbf{p}):=e^{-ip  \hat{x} }\otimes a(\mathbf{p})= 
 	V_{\xi} a_{\Theta}(\mathbf{p})V^{*}_{\xi},
 	\end{align}
 	where an analogous relation holds for the creation operator. 
 	\\\\
 	In this section we define the space of the non-abelian Moyal-Weyl plane and define a QFT on it. Furthermore, we deform the quantum field with  non-abelian warped convolutions and finally show the existence of an isomorphism between the two algebras.
 	Hence, we start with the following lemma. 
 	\begin{lemma}\label{namw} Let $\mathcal{V}_{\tau}$ denote the representation space of the *-algebra of the non-abelian Moyal-Weyl plane  generated by the self-adjoint operators ${X}^{\Theta_{\tau}}_{\mu}$, $\mu=0,1,\cdots,n$, given on the dense domain $ \mathscr{S}(\mathbb{R}^d)\otimes{\mathbb{C}^{m}}$,   that fulfill   
 		\begin{equation}\label{xcr1}
 		[{X}^{\Theta_{\tau}}_{\mu},{X}^{\Theta_{\tau}}_{\nu}]=2i\Theta_{\mu\nu}\otimes Y_{\alpha}\tau^{\alpha}.
 		\end{equation}
 		Then $\Theta_{\mu\nu}Y^{\alpha}\tau_{\alpha}$ is the center of the algebra.
 	\end{lemma} 
 	
 	\begin{proof} 
 		The proof is straight forward,
 		\begin{align*}[{X}^{\Theta_{\tau}}_{\mu},\Theta_{\rho\nu}\otimes Y_{\beta}\tau^{\beta}] &=
 		[ {X}_{\mu}+   (\Theta P)_{\mu}\otimes Y_{\alpha} \tau^ {\alpha},\Theta_{\rho\nu}\otimes Y_{\beta}\tau^{\beta}] \\&=
 		[  (\Theta P)_{\mu}\otimes Y_{\alpha} \tau^ {\alpha},\Theta_{\rho\nu}\otimes Y_{\beta}\tau^{\beta}] \\&=[  (\Theta P)_{\mu},\Theta_{\rho\nu}]\otimes Y_{\alpha} \tau^ {\alpha} Y_{\beta}\tau^{\beta}=
 		0 .
 		\end{align*}
 	\end{proof}
 	\begin{remark} 
 		Note that we increased the dimension of the Heisenberg-Weyl algebra by one and  the operator $P_{\mu}$ should not be confused with the quantum field theoretical momentum operator. Moreover, we changed the sign of $\Theta$ to keep the known notion of the Moyal-Weyl plane. 
 	\end{remark}
 	Next we define the creation and annihilation operators of the quantum fields defined on the non-abelian Moyal-Weyl plane in a natural manner as (see \cite{DFR}),
 	\begin{align} a_{\otimes_{\tau}}(\mathbf{p}):&=e^{ -ip   {X}^{{\Theta_{\tau}}} }\otimes (a(\mathbf{p})\otimes \mathbb{I}_{m\times m} ) ,\\
 	a^*_{\otimes_{\tau}}(\mathbf{p}):&=e^{  ip   {X}^{{\Theta_{\tau}}} }\otimes  (a^*(\mathbf{p})\otimes \mathbb{I}_{m\times m} ).
 	\end{align}
 	The quantum fields are defined on the tensor product space  $\mathcal{V}_{\tau}\otimes\mathcal{H}\otimes{\mathbb{C}^{m}}$. In order to give the correspondence of   fields on  $\mathcal{H}\otimes{\mathbb{C}^{m}}$, which can be obtained by a non-abelian deformation by using the QFT  momentum, we shall construct the unitary operator  $V^{\tau}_{\xi}$. However, to do so we first deform the annihilation and creation operator.
 	\begin{lemma}\label{nacaar}
 		The  annihilation and creation operator  deformed by using the non-abelian-warped convolutions, where the generator  $Q(\mathbf{P})_{\mu}=-P_{\mu}\otimes Y_{\alpha}  \tau^{\alpha}$ is used, with $P_{\mu}$ being the field theoretical momentum operator, are well-defined and  given on $\left(\bigotimes_{i=1}^{k} \mathscr{S}(\mathbb{R}^n)\right)\otimes{\mathbb{C}^{m}}$  as follows
 		\begin{align} 
 		a_{\Theta_{\tau}} (\mathbf{p}) &=
 		e^{ i    p_{\mu}(\Theta P)^{\mu} \otimes Y^{ \alpha} \tau_{\alpha}}(a(\mathbf{p})\otimes \mathbb{I}_{m\times m} ),\\
 		a^*_{\Theta_{\tau}} (\mathbf{p}) &=
 		e^{-i p_{\mu}(\Theta P)^{\mu} \otimes Y^{ \alpha} \tau_{\alpha}} (a^*(\mathbf{p})\otimes \mathbb{I}_{m\times m} ).
 		\end{align} 
 	\end{lemma}
 	\begin{proof}To prove that the expression is well-defined note that for $\Psi_k \in \left(\bigotimes_{i=1}^{k} \mathscr{S}(\mathbb{R}^n)\right)\otimes{\mathbb{C}^{m}}$ we have
 		\begin{align} \label{ineq1}
 		\|
 		a_{\Theta_{\tau}} (\mathbf{p})\Psi_k\|=\|
 		e^{ i (\Theta P)^{\mu} p_{\mu}\otimes Y^{ \alpha} \tau_{\alpha}}(a(\mathbf{p})\otimes \mathbb{I}_{m\times m} )\Psi_k\|=
 		\| (a(\mathbf{p})\otimes \mathbb{I}_{m\times m} )\Psi_k\|.
 		\end{align} 
 		Therefore, equivalent inequalities as in the free case hold  for non-abelian-deformed quantum field operators. Next, to calculate the non-abelian-deformation with the momentum operator let us examine the following expression first. 
 		\begin{align*} 
 		\alpha^{\tau}_{k}( a(\mathbf{p})\otimes \mathbb{I}_{m\times m}  )&=e^{-i k^{\mu} P_{\mu}\otimes Y^{ \alpha} \tau_{\alpha}}(a(\mathbf{p})\otimes \mathbb{I}_{m\times m} )e^{ i k^{\mu} P_{\mu}\otimes Y^{ \alpha} \tau_{\alpha}}\\&=e^{ i k^{\mu} p_{\mu}\otimes Y^{ \alpha} \tau_{\alpha}}(a(\mathbf{p})\otimes \mathbb{I}_{m\times m} ),
 		\end{align*} 
 		where this expression can be calculated by using Baker-Hausdorff formula or the same technical proof as done in \cite{Sch} for the abelian (usual) case. Hence by applying the adjoint action we obtain
 		\begin{align*} 
 		a_{\Theta_{\tau}} (\mathbf{p}) \Psi_k&=   \int \alpha^{\tau}_{\Theta k}(a(\mathbf{p})\otimes \mathbb{I}_{m\times m} )\, dE(k)\Psi_k\\&=
 		\int e^{ i (\Theta k)^{\mu} p_{\mu}\otimes Y^{ \alpha} \tau_{\alpha}}(a(\mathbf{p})\otimes \mathbb{I}_{m\times m} )\,dE(k)\Psi_k\\&=
 		e^{ i (\Theta P)^{\mu} p_{\mu}\otimes Y^{ \alpha} \tau_{\alpha}}(a(\mathbf{p})\otimes \mathbb{I}_{m\times m} )\Psi_k,
 		\end{align*} 
 		where the proof for the creation operator can be done analogously. 
 	\end{proof}

 	\begin{theorem}
 		An equivalence  exists between the $*$-algebras of fields deformed with non-abelian warped convolutions using the momentum operator and the   $*$-algebra of the free fields on the non-abelian Moyal-Weyl plane $\mathcal{V}_{\tau}$. This isomorphism is given by  the unitary  operator $V^{\tau}_{\xi}= \bigoplus_{n=0}^{\infty}V^{ \tau ,(n)}_{\xi}:\mathcal{H}\otimes{\mathbb{C}^{m}}\rightarrow\mathcal{V}_{\tau}\otimes\mathcal{H}\otimes{\mathbb{C}^{m}}$,  with $\xi\in\mathcal{V}_{\tau}$ such that $||\xi||_{\mathcal{V}_{\tau}}=1$, $\Psi_n\in\mathcal{H}_n$ and a non-zero vector $e_{\alpha}\in{\mathbb{C}^{m}}$
 		\begin{equation}\label{nc1}
 		\left(V^{\tau ,(n)}_{\xi}\Psi_n\otimes e_{\alpha}\right)\left(\left(\mathbf{p}_1,\dots,\mathbf{p}_n\right)\otimes 1\right)=\left(\Psi_n\left(\mathbf{p}_1,\dots,\mathbf{p}_n\right)\otimes e_{\alpha}\right)\cdot e^{i\sum\limits_{k=1}^{n}p_{k}{X}^{{\Theta_{\tau}}}}\xi.
 		\end{equation}
 		Hence,   the following equations hold in a distributional sense 
 		\begin{align}\label{aot} a_{\otimes_{\tau}}(\mathbf{p}):=e^{-ip   {X}^{{\Theta_{\tau}}} }\otimes (a(\mathbf{p})\otimes \mathbb{I}_{m\times m} )= 
 		V^{\tau}_{\xi} a_{\Theta_{\tau}}(\mathbf{p})V^{ *{\tau}}_{\xi},
 		\end{align}
 	\end{theorem}
 	\begin{proof}
 		The proof is straight forward by using Lemma \ref{namw} and Lemma \ref{nacaar}.
 		\begin{align*} 
 		&\frac{1}{\sqrt{n+1}} \left( V^{\tau}_{\xi} a_{\Theta_{\tau}}(\mathbf{p})V^{ *{\tau}}_{\xi}\Psi_n\otimes e_{\alpha}\right)\left(\left(\mathbf{p}_1,\dots,\mathbf{p}_n\right)\otimes 1\right)\\=& \frac{1}{\sqrt{n+1}}e^{ i (\Theta \sum\limits_{k=1}^{n}p_{k})^{\mu} p_{\mu}\otimes Y^{ \alpha} \tau_{\alpha}} \left(a(p)V^{ *{\tau}}_{\xi}\Psi_n\otimes e_{\alpha}\right)\left(\left(\mathbf{p}_1,\dots,\mathbf{p}_n\right)\otimes 1\right)\cdot e^{i\sum\limits_{k=1}^{n}p_{k}{X}^{{\Theta_{\tau}}}}\xi
 		\\ = & e^{ i (\Theta \sum\limits_{k=1}^{n}p_{k})^{\mu} p_{\mu}\otimes Y^{ \alpha} \tau_{\alpha}} \left( V^{ *{\tau}}_{\xi}\Psi_{n+1}\otimes e_{\alpha}\right)\left(\left(\mathbf{p} ,\mathbf{p}_1,\dots,\mathbf{p}_n\right)\otimes 1\right)\cdot e^{i\sum\limits_{k=1}^{n}p_{k}{X}^{{\Theta_{\tau}}}}\xi   
 		\\=&   e^{ i (\Theta \sum\limits_{k=1}^{n}p_{k})^{\mu} p_{\mu}\otimes Y^{ \alpha} \tau_{\alpha}} \left(  \Psi_{n+1}\otimes e_{\alpha}\right)\left(\left(\mathbf{p} ,\mathbf{p}_1,\dots,\mathbf{p}_n\right)\otimes 1\right)\cdot e^{-i(\sum\limits_{k=1}^{n}p_{k}+ p){X}^{{\Theta_{\tau}}}}e^{i\sum\limits_{k=1}^{n}p_{k}{X}^{{\Theta_{\tau}}}}   
 		\\=&    \frac{1}{\sqrt{n+1}}  \left(  a(p)\Psi_{n}\otimes e_{\alpha}\right)\left(\left( \mathbf{p}_1,\dots,\mathbf{p}_n\right)\otimes 1\right)\cdot 
 		e^{-ip{X}^{{\Theta_{\tau}}}} ,
 		\end{align*} 
 		where in the last lines we used the Baker-Hausdorff formula, the commutation relations of the non-abelian Moyal-Weyl and the center of the algebra.
 	\end{proof}
 	
 	\begin{definition}  \label{defqfx}
 		Let $\Theta$ be a real skew-symmetric matrix w.r.t. the Lorentzian scalar-product on
 		$\mathbb{R}^{d}$ and  
 		let $\chi \in\mathscr{S}(\mathbb{R}^d\times\mathbb{R}^d)$ with $\chi(0,0)=1$.  Furthermore, let
 		$\phi(f)$ be the massive free scalar 
 		field smeared out with functions $f \in \mathscr{S}(\mathbb{R}^d)$. Then,
 		the operator valued distribution $\phi(f)$ deformed with the    operator $Q(\mathbf{P})_{\mu}=-P_{\mu}\otimes Y_{\alpha}  \tau^{\alpha}$  
 		denoted as $\phi_{\Theta_{\tau}}(f)$,  is defined on vectors of the dense domain 
 		$\left(\bigotimes_{i=1}^{k} \mathscr{S}(\mathbb{R}^n)\right)\otimes{\mathbb{C}^{m}}$  as follows
 		\begin{align}\label{defqfx2}
 		\phi_{\Theta_{\tau}}(f)\Psi_{k}:&=(2\pi)^{-d}
 		\lim_{\epsilon\rightarrow 0}
 		\iint  \, dx \,  dy \, e^{-ixy}  \, \chi(\epsilon x,\epsilon y)\alpha^{\tau}_{\Theta x} 
 		(\phi(f))U({y})\Psi_{k}\nonumber\\&=
 		(2\pi)^{-d}\lim_{\epsilon\rightarrow 0}
 		\iint  \, dx \,  dy \, e^{-ixy}  \, \chi(\epsilon x,\epsilon y) \alpha^{\tau}_{\Theta x} 
 		\left(a(\overline{f^-})+a^*(f^+)\right)U({y})\Psi_{k}\nonumber\\&
 		=:\left(a_{\Theta_{\tau}}(\overline{f^-})+a_{\Theta_{\tau}}^*(f^+)\right)\Psi_{k}.
 		\end{align}
 		The automorphism $\alpha^{\tau}$ is defined by the adjoint action of the unitary operator $U^{\tau}(y)$ and the
 		test functions $f^{\pm}(\mathbf{p})$ in momentum space are defined as follows
 		\begin{equation}\label{tf}
 		f^{\pm}(\mathbf{p}):=\int dx \,f(x)e^{\pm ipx}, \qquad 
 		p=(\omega_{\mathbf{p}},\mathbf{p})\in H_m^{+}.
 		\end{equation}
 	\end{definition}

 	\begin{lemma} \label{ldx2}
 		For  $\Psi_k \in \left(\bigotimes_{i=1}^{k} \mathscr{S}(\mathbb{R}^n)\right)\otimes{\mathbb{C}^{m}}$  the familiar 
 		bounds of  the free field hold for the deformed field 
 		$\phi_{\Theta_{\tau}} (f)$ and therefore the deformation with  operator $Q(\mathbf{P})_{\mu}=-P_{\mu}\otimes Y_{\alpha}  \tau^{\alpha}$  
 		is well-defined.
 	\end{lemma} 
 	\begin{proof}  
 		By using Lemma \ref{nacaar} one obtains the familiar bounds for a free 
 		scalar field. 
 		\begin{align*} \left\Vert\phi_{\Theta_{\tau}} (f)\Psi_k\right\Vert& \leq 
 		\left\Vert   a( \overline{ f^-})_{\Theta_{\tau}}\Psi_k \right\Vert
 		+\left\Vert a^*(  f^+)_{\Theta_{\tau}}\Psi_k \right\Vert \\& \leq 
 		\left\Vert   f^{+} \right\Vert  \left\Vert (N+1)^{1/2 }\Psi_k 
 		\right\Vert
 		+ \left\Vert   f^{-}\right\Vert   \left\Vert (N+1)^{1/2 
 		}\Psi_k\right\Vert 
 		\end{align*}
 		where in the last lines we used the triangle inequality, the Cauchy-Schwarz inequality and the 
 		bounds given in  Inequality \ref{nacaar}. 
 	\end{proof} 
 	\subsection{Wightman Properties of the Deformed QF}
 	In this section we show that the deformed field $\phi_{\Theta_{\tau}}(f)$
 	satisfies the Wightman properties with the exception of covariance and 
 	locality. This is the subject of the following proposition.  We  use $\mathcal{H}$ for the Bosonic Fockspace.\newline
 	\begin{proposition}\label{prop1x}
 		Let $\Theta$ be a real skew-symmetric matrix w.r.t. the Lorentzian scalar-product on
 		$\mathbb{R}^{d}$ and $f \in \mathscr{S}(\mathbb{R}^d)$.\\
 		\begin{itemize}
 			\item[a)] The dense subspace $\mathcal{D} \otimes{\mathbb{C}^{m}}$ of vectors of finite 
 			particle number is contained in the domain 
 			$\mathcal{D}^{\Theta_{\tau}}=\{\Psi\in \mathcal{H} \otimes \mathbb{C}^{m}|  \left\Vert 
 			\phi_{\Theta_{\tau}}(f)\Psi\right\Vert^2 < \infty \}$ of any 
 			$\phi_{\Theta_{\tau}}(f)$. Moreover, 
 			$\phi_{\Theta_{\tau}}(f)(\mathcal{D}\otimes{\mathbb{C}^{m}})\subset\mathcal{D}\otimes{\mathbb{C}^{m}}$ and 
 			$\phi_{\Theta_{\tau}}(f)(\Omega\otimes e_{\alpha})=\phi(f)(\Omega\otimes e_{\alpha})$.\\
 			\item[b)] For scalar fields deformed via warped convolutions and 
 			$\Psi\in\mathcal{D} \otimes{\mathbb{C}^{m}}$,
 			\begin{equation*}
 			f\longmapsto\phi_{\Theta_{\tau}}(f)\Psi
 			\end{equation*}
 			is a vector valued tempered distribution.\\
 			\item[c)]
 			For $\Psi\in\mathcal{D} \otimes{\mathbb{C}^{m}}$  and $\phi_{\Theta_{\tau}}(f)$  the following holds
 			\begin{equation*}  \phi_{\Theta_{\tau}}(f)^{*}\Psi =  
 			\phi_{\Theta_{\tau}}(\overline{f })\Psi.
 			\end{equation*} 
 			For real $f\in\mathscr{S}(\mathbb{R}^{d})$, the deformed field 
 			$\phi_{\Theta_{\tau}}(f)$ is essentially self-adjoint on
 			$\mathcal{D} \otimes{\mathbb{C}^{m}}$.\\
 			\item[d)] The Reeh-Schlieder property holds: Given an  open set of 
 			space-time $\mathcal{O}\subset \mathbb{R}^d$, then
 			\begin{equation*}
 			\mathcal{D}^{\Theta_{\tau}}(\mathcal{O}):=  
 			span\{\phi_{\Theta_{\tau}}(f_1)\dots\phi_{\Theta_{\tau}}(f_k)(\Omega\otimes e_{\alpha}): k\in \mathbb{N}, 
 			f_1\dots f_k\in \mathscr{S}(\mathcal{O})\}
 			\end{equation*}
 			is dense in $\mathcal{H} \otimes{\mathbb{C}^{m}}$.
 		\end{itemize}
 	\end{proposition}
 	\begin{proof}  
 		a) The fact that $\mathcal{D}\otimes{\mathbb{C}^{m}}\subset \mathcal{D}^{\Theta_{\tau}}$,  follows 
 from Lemma \ref{ldx2} , since the deformed scalar field is bounded by the bounds of the free field. Moreover, the equivalence of the	deformed field acting on the vacuum and the undeformed field acting  
 		on $\Omega\otimes e_{\alpha}$, is  due to the property of the unitary 
 		operators $U^{\tau}( {p})(\Omega\otimes e_{\alpha})=\Omega\otimes e_{\alpha}$.
 		\\\\
 		b) The use of Lemma \ref{ldx2} implies that the right hand side 
 		depends continuously on the function $f$, hence the temperateness of 
 		$f\longmapsto\phi_{\Theta_{\tau}}(f)\Psi$,  $\Psi\in\mathcal{D}\otimes{\mathbb{C}^{m}}$ follows.\\\\
 		c) Hermiticity of the deformed field  $\phi_{\Theta_{\tau}}(f)$ is first proven. This is done
 	 as the proof of Lemma \ref{wcl2}, demonstrating hermiticity of a deformed
 		operator if the undeformed one is self-adjoint, 
 		\begin{align*}
 		\phi_{\Theta_{\tau}}(f)^{*}\Psi&=  
 		\phi_{\Theta_{\tau}}(\overline{f })\Psi.
 		\end{align*}  For real  		$f$ the essential  self-adjointness of the hermitian deformed field $\phi_{\Theta_{\tau}} (f)$ is proven by showing that the
 		deformed field has a dense set of analytic vectors. Next, by  
 		Nelson's analytic vector theorem, it  follows that the deformed field
 		$\phi_{\Theta_{\tau}}(f)$ is essentially self-adjoint on this dense set of analytic vectors, (for
 		similar proof see \cite[Chapter I, Proposition 5.2.3]{BR}). \newline\newline  For $\Psi_{k} \in
 		\mathcal{H}_{k}\otimes{\mathbb{C}^{m}}$ the estimates of the $l$-power
 		of the deformed field $\phi_{\Theta_{\tau}}(f)$, are given in the following
 		\begin{equation*}
 		\left\Vert
 		\phi_{\Theta_{\tau}}(f)^l\Psi_{k}\right\Vert\leq
 		2^{l/2}(k+l)^{1/2}(k+l-1)^{1/2}\cdots(k+1)^{1/2}\left\Vert f\right\Vert^{l}\left\Vert \Psi_{k}
 		\right\Vert,
 		\end{equation*}
 		where in the last lines we used Lemma \ref{ldx2} for the estimates of the deformed field. Finally,
 		we can write the sum
 		
 		\begin{equation*}
 		\sum\limits_{l\geq 0} \frac{\vert t\vert^l}{l!}\left\Vert
 		\phi(f)^l \Psi_k
 		\right\Vert
 		\leq
 		\sum\limits_{l \geq 0} \frac{(\sqrt{2}|t|)^l}{l!}\left(
 		\frac{(k+l)!}{k!}\right)^{1/2}\left\Vert f\right\Vert^{l}\left\Vert \Psi_{k} \right\Vert<
 		\infty
 		\end{equation*}
 		for all $t\in \mathbb{C}$. It follows that each $\Psi\in \mathcal{D}\otimes{\mathbb{C}^{m}}$ is an analytic vector for
 		the deformed field $\phi_{\Theta_{\tau}}(f)$. Since the set $\mathcal{D}\otimes{\mathbb{C}^{m}}$ is dense in
 		$\mathcal{H}\otimes{\mathbb{C}^{m}} $,  Nelson's analytic vector theorem implies that $\phi_{\Theta_{\tau}}(f)$
 		is essentially self-adjoint on $\mathcal{D}\otimes{\mathbb{C}^{m}}$. 
 		\\\\
 		d)  
 		The properties of the  unitary operator  $U( {y})$
 		of translations  leads to the application of the 
 		standard Reeh-Schlieder argument \cite{SW} which states that   
 		$\mathcal{D}^{\Theta_{\tau}}\otimes{\mathbb{C}^{m}}$ is dense in $\mathcal{H}\otimes{\mathbb{C}^{m}}$ if and 
 		only if $\mathcal{D}^{\Theta_{\tau}}\otimes{\mathbb{C}^{m}}$ is dense in 
 		$\mathcal{H}\otimes{\mathbb{C}^{m}} $.  We choose the functions  $f_1,\dots,f_k \in \mathscr 
 		{S}(\mathbb{R}^{d})$ such that the Fourier transforms of the functions 
 		do not intersect the lower mass shell and therefore the domain 
 		$\mathcal{D}^{\Theta_{\tau}} \otimes{\mathbb{C}^{m}}$ consists of the following vectors
 		\begin{align*} 
 		\phi_{\Theta_{\tau}}(f_1)\dots
 		\phi_{\Theta_{\tau}}(f_k)(\Omega\otimes e_{\alpha})&=
 		a_{\Theta_{\tau}}^*(f^{+}_1)\dots
 		a_{\Theta_{\tau}}^*(f^{+}_k)(\Omega\otimes e_{\alpha}) 
 		\\&=     
 		\sqrt{m!}P_{m}(S^{\Theta_{\tau}}_k(  f^{+}_1\otimes \dots \otimes 
 		f^{+}_k)) ,
 		\end{align*}
 		where $P_k$ denotes the orthogonal projection from 
 		$\mathcal{H}_1^{\otimes k}$ onto its totally symmetric subspace 
 		$\mathcal{H}_{k}$, and $S^{\Theta_{\tau}}_k \in \mathscr{B}(\mathcal{H}_1^{\otimes 
 			k})\otimes{\mathbb{C}^{m}}$ is a multiplication operator-valued unitary matrix given as
 		\begin{align*} S^{\Theta_{\tau}}_k(p_1,\dots,p_k)=\prod \limits_{1\leq l < j \leq k} 
 		e^{{i}p_l \Theta p_j\otimes Y^{ \alpha}  \tau_{\alpha}  }.
 		\end{align*}
 		Following the 
 		same arguments as in \cite{GL1} the density of $\mathcal{D}^{\Theta_{\tau}}({\mathbb{R}^d})\otimes{\mathbb{C}^{m}}$ in $\mathcal{H}\otimes{\mathbb{C}^{m}}$ follows.  
 	\end{proof}

 	\subsection{Wedge-Covariance and Wedge-Locality}\label{s32}
 	In \cite{GL1},  a map was constructed 
 	$Q:W \mapsto Q(W)$ from a 
 	set $\mathcal{W}_{0}:=\mathcal{L}^{\uparrow}_{+}W_{1}$ of wedges, where 
 	$W_{1}:=\{x\in \mathbb {R}^d: x_1 > |x_0| \}$  to a set 
 	$\mathcal{Q}_{0}\subset \mathbb{R}^{-}_{d\times d}$ of skew-symmetric 
 	matrices. The corresponding fields are given by
 	$\phi_{W}(x):=\phi(Q(W),x)$. \\\\ 
 	Hence, the deformed scalar field $\phi(Q(W),x)$  can be given as a  field defined on the wedge. The homomorphism $Q:W \mapsto Q(W)$   is given by the following,

 	\begin{definition}Let  $\Theta$ be a real skew-symmetric matrix on 
 		$\mathbb{R}^d$ then the map $\gamma_{\Lambda}(\Theta)$
 		is defined as follows
 		\begin{align}\label{hm}
 		\gamma_{\Lambda}(\Theta):=
 		\left\{
 		\begin{array} {cc}
 		\Lambda\Theta\Lambda^T, \qquad &\Lambda\in\mathcal{L}^{\uparrow}, \\ 
 		-\Lambda\Theta\Lambda^T,\qquad &\Lambda\in\mathcal{L}^{\downarrow} .
 		\end{array} \right.
 		\end{align}
 	\end{definition}
 	\begin{definition}
 		$\Theta$ is called an admissible matrix if the realization of the 
 		homomorphism $Q(\Lambda W)$ defined by the map 
 		$\gamma_{\Lambda}(\Theta)$   is a well defined mapping. This is the case 
 		iff $\Theta$ has in $d$ dimensions the following form
 		\begin{align}
 		\begin{pmatrix} 0 & \lambda & 0 & \cdots & 0 \\ \lambda& 0 & 0 & \cdots 
 		& 0  \\ 0 & 0 &0 &\cdots&0
 		\\ \vdots & \vdots &\vdots & \ddots& \vdots\\ 0 & 0 &0 & \cdots&0
 		\end{pmatrix},\qquad \lambda\geq0.
 		\end{align}
 		For the physical most interesting case of 4 dimensions the 
 		skew-symmetric matrix $\Theta$ has the more general form
 		\begin{align}\qquad\quad
 		\begin{pmatrix} 0 & \lambda & 0 & 0 \\ \lambda& 0 & 0 & 0  \\ 0 & 0 &0 &\eta
 		\\ 0 & 0 &-\eta &0
 		\end{pmatrix},\qquad \lambda\geq0, \eta \in \mathbb{R}.
 		\end{align}
 		
 	\end{definition} 
 	By using the former definitions we give the following correspondence of the fields on $\mathcal{H}\otimes{\mathbb{C}^{m}}$ as,
 	\begin{equation}\label{ex1}
 	\phi_W(f)\Psi:=\phi(Q(W),f)\Psi=\phi_{\Theta_{\tau}}( f)\Psi. 
 	\end{equation}Next, we define the covariance and locality properties of the wedge-fields.   In particular, we write the definitions of a wedge-covariant and a wedge-local field, (\cite{GL1}, Definition 3.2).
 	\newline
 	\begin{definition} \label{drca}
 		Let $\phi=\{\phi_W: W\in\mathcal{W}_{0}\}$ denote the family of fields
 		satisfying the domain and continuity assumptions of the Wightman axioms. Then, the field  $\phi $
 		is defined to be a wedge-local quantum field if the following  
 		conditions are satisfied: 
 		\\ 
 		\begin{itemize}
 			\item \textbf{Covariance:}  For any $W \in\mathcal{W}_{0} $ and 
 			$f\in\mathscr{S}(\mathbb{R}^d)$ the following holds
 			\begin{align*}
 			U( y, \Lambda)\phi_W(f)U( y, \Lambda)^{-1}&=\phi_{\Lambda W}(f\circ( y,
 			\Lambda)^{-1} ),\qquad  (y, \Lambda)\in \mathcal{P}^{\uparrow}_{+},
 			\\
 			U(0, j)\phi_W(f) U( 0,j)^{-1}&=\phi_{ jW}(\overline{f}\circ (0, j)^{-1}) ,
 			\end{align*}$\,$  where $j$ represents the space-time reflections, i.e. $x^{\mu}\rightarrow -x^{\mu}.$\\
 			\item  \textbf{Wedge-locality:}  Let $W,\tilde{W}\in\mathcal{W}_{0}$ and
 			$f\in\mathscr{S}(\mathbb{R}^2)$. If 
 			\begin{equation*}
 			\overline{W+\text{supp } f}\subset
 			(\tilde{W}+\text{supp } g)',
 			\end{equation*}  \end{itemize}
 		then 
 		\begin{equation*}
 		[\phi_{W}(f), \phi_{\tilde{W}}(g)]\Psi=0,\quad \Psi \in \mathcal{D}.
 		\end{equation*}

 	\end{definition}
 By using the geometrical properties of the wedge the following lemma results, (\cite{GL1}, Lemma 3.3).\newline
 	\begin{lemma}
 		Let $\phi=\{\phi_W: W\in\mathcal{W}_{0}\}$ denote the family of fields
 		satisfying the domain, continuity and covariance assumptions stated in Definition \ref{drca}. Then $\phi$ is
 		wedge-local if and only if
 		\begin{equation*}
 		[\phi_{W_1}(f), \phi_{-{W}_1}(g)]\Psi=0,\quad \Psi \in \mathcal{D},
 		\end{equation*}
 		for all $f,g \in C_{0}^{\infty}(\mathbb{R}^d)$ with  $\text{supp } f\subset W_{1}$ and 
 		$\text{supp } g\subset -W_{1}$.
 	\end{lemma}$\,$ \\So let us first investigate the wedge-covariance properties of our non-abelian fields. The result is given in the following proposition. 
 	\begin{proposition}\label{wl0} The deformed fields $\phi_{\Theta_{\tau}}(    f)  $ transform  under the adjoint action of the proper orthochronous Poincar\'e group as follows,
 		\begin{align*} 
 		U( x, \Lambda) \phi_{\Theta_{\tau}}(    f)   U( x, \Lambda)^{-1}&=\phi_{(\Lambda\Theta\Lambda^T)_{\tau}   } (    f\circ  ( x,
 		\Lambda)^{-1}),\qquad  (y, \Lambda)\in \mathcal{P}^{\uparrow}_{+},\\
 		U(0, j)\phi_{\Theta_{\tau}}(    f)  U( 0,j)^{-1}&=\phi_{ (-\Lambda\Theta\Lambda^T)_{\tau}}(\overline{f}\circ (0, j)^{-1}) .
 		\end{align*} 
 		Therefore, the field $\phi_{\Theta_{\tau}}$ is a  wedge-covariant field. 
 	\end{proposition}
 	\begin{proof} 
 		The proof is done straight forward  by applying the unitary operator of the Poincar\'e group on the field and by taking the transformation of creation- and annihilation operators into account. \begin{align*} &
 		U( x, \Lambda) \phi_{\Theta_{\tau}}(    f)   U( x, \Lambda)^{-1}\\&=(2\pi)^{-d} 
 		\lim_{\epsilon\rightarrow 0}
 		\iint  \, dy \,  du \, e^{-iyu}  \, \chi(\epsilon y,\epsilon u)  U( x, \Lambda)\alpha^{\tau}_{{\Theta y}
 		}(\phi(f))U({u}) \,  U( x, \Lambda)^{-1}\\&=(2\pi)^{-d} 
 		\lim_{\epsilon\rightarrow 0}
 		\iint  \, dy \,  du \, e^{-iyu}  \, \chi(\epsilon y,\epsilon u) \alpha^{\tau}_{\Lambda\Theta\Lambda^{T} y} 
 		(\phi(f\circ( x,
 		\Lambda)^{-1}))    
 		\\&=  \phi_{(\Lambda\Theta\Lambda^T)_{\tau}  } (    f\circ  ( x,
 		\Lambda)^{-1}) .
 		\end{align*}
 		
 	\end{proof}
 	Let us examine the proof of wedge-locality for the abelian deformed fields, \cite{GL1}. The proofs works  for   smearing functions that are entire analytic  and therefore an analytical  continuation to the complex upper half plane can be done. The following coordinate transformation will ease the proof and it is given by,
 	\begin{equation*}
 	m_{\perp}:=(m^2+p_{\perp})^{1/2}, \qquad p_{\perp}:=(p_2,\dots,p_{n}),\qquad \vartheta:=\text{arc}\sinh\frac{p_1}{m_{\perp}}.
 	\end{equation*}
 	 The measure and  the on-shell momentum vector is written in the  new coordinates as follows, 
 	\begin{equation*}
 	d^n\mu(\mathbf{p})= d^{n-1}p_{\perp}d\vartheta,\qquad \qquad p(\vartheta):=\left(
 	\begin{array}{c}
 	m_{\perp}\cosh \vartheta \\
 	m_{\perp}\sinh \vartheta  \\
 	p_{\perp}
 	\\
 	\end{array}
 	\right)
 	\end{equation*}
 	By the use of the coordinate transformation and the analyticity of the function  one obtains for the smeared functions $f\in
 	C_{0}^{\infty}(W_{1})$ and $g\in C_{0}^{\infty}(-W_{1})$, (see Equation \ref{tf})
 	\begin{equation}\label{ac1}
 	f^{-}(p_{\perp}, \vartheta+i\pi)=
 	f^{+}(-p_{\perp},\vartheta) ,\qquad  g^{-}(p_{\perp},  \vartheta+i\pi)=
 	g^{+}(-p_{\perp},\vartheta) .
 	\end{equation}
 	Now for the proof of wedge-locality in our case the same arguments are applied. 
 	\begin{proposition}
 		The family of fields $\phi=\{\phi_W: 
 		W\in\mathcal{W}_{0}\}$ defined by $\phi_W(f):=\phi(Q(W),f)=\phi_{\Theta_{\tau}}( f)$
 		for $\tau\in su(2)$, with $Y^{\alpha}$ being a vector, are \textbf{not} wedge-local fields on the Bosonic Fockspaces $\mathscr{H} \otimes\mathbb{C}^2$.  However, by choosing $Y$ to be matrix valued such that $Y_{\alpha}\tau^{\alpha}$ has positive eigenvalues the family of fields $\phi$  are \textbf{wedge-local}  on the Bosonic Fockspace $\mathscr{H} \otimes\mathbb{C}^m$.
 	\end{proposition}

 	\begin{proof}
 		For the proof  we use Lemma \ref{wcl2} and the proof of wedge-locality for a free translated scalar field,  
 		\cite{GL1}.
 		In order to use Lemma \ref{wcl2}, we have to show that the following commutator vanishes for $f\in
 		C_{0}^{\infty}(W_{1})$ and $g\in C_{0}^{\infty}(-W_{1})$,
 		\begin{align*}
 		[\alpha^{\tau}_{\Theta x}(\phi(f)),\alpha^{\tau}_{-\Theta y}(\phi(g))]&=
 		[\alpha^{\tau}_{\Theta x}(a(\overline{f^-})),\alpha^{\tau}_{-\Theta y}(a^*({g^+})]-
 		[\alpha^{\tau}_{-\Theta y}(a(\overline{g^-})),\alpha^{\tau}_{\Theta x}(a^*({f^+})] 
 		,
 		\end{align*}
 		where in the former lines all other terms are equal to zero and the unitary equivalence was used.  
 		Let us first take a look at the first expression of the
 		commutator, 
 		\begin{align*} 
 		&
 		[\alpha^{\tau}_{\Theta x}\bigl(
 		a(\overline{f^-})
 		\bigr),\alpha^{\tau}_{-\Theta y}\bigl(
 		a^*({g^+}) 
 		\bigr)
 		] 
 		\\&
 		=
 		\int d^n\mu(\mathbf{p})\int d^n\mu(\mathbf{k})
 		f^{-} (\mathbf{p})  g^{+}  (\mathbf{k}) e^{ ip\Theta x
 			\otimes Y^{ \alpha} \tau_{\alpha}} e^{ ik\Theta y \otimes
 			Y^{ \beta} \tau_{\beta}} 
 		[  a (\mathbf{p}) ,  a^* (\mathbf{k}) ]  \otimes \mathbb{I}_{m\times m}  
 		\\& 
 		=
 		\int d^n\mu(\mathbf{p}) 
 		f^{-} (\mathbf{p}) g^{+} (\mathbf{p}) e^{ ip\Theta (x+y)\otimes Y^{ \alpha} \tau_{\alpha} }\\&
 		=
 		\int  d^{n-1}p_{\perp}d\vartheta\,
 		f^{-} (p_{\perp},\vartheta)  g^{+} (p_{\perp},\vartheta)e^{ ip(\vartheta)\Theta (x+y)\otimes Y^{ \alpha} \tau_{\alpha} }\\&
 		=
 		\int  d^{n-1}p_{\perp}d\vartheta\,
 		f^{+} (-p_{\perp},\vartheta)  g^{-}(-p_{\perp},\vartheta)e^{ ip(\vartheta+i\pi)\Theta (x+y)\otimes Y^{ \alpha} \tau_{\alpha} } 
 		\\&
 		=\sum_{r=1}^{m}
 		\int  d^{n-1}p_{\perp}d\vartheta\,
 		f^{+} (-p_{\perp},\vartheta)  g^{-}(-p_{\perp},\vartheta)e^{ i\lambda_r p(\vartheta+i\pi)\Theta (x+y)} \otimes W B_r W^{-1}
 		,
 		\end{align*}
 		where in the last lines we used  the boundedness and analyticity properties of the unitary transformed functions $f,g$ (see \cite[Proposition 3.4]{GL1}). However this cannot be done if  $|e^{i\lambda_r p(\vartheta+i\pi)\theta (x+y)} |\nleq 1,    \forall r$, since it would be unbounded and hence we cannot shift the contour of the integral from $\mathbb{R}$ to $\mathbb{R}+i\pi$.  Unboundedness in the case where $Y\in\mathbb{R}^3$   is seen by calculating the eigenvalues of the matrix $Y^{ \alpha} \tau_{\alpha}$, which are for $su(2)$, $\pm  (Y^{ \alpha}Y_{ \alpha})^{1/2}$. Therefore, it is easy to realize that it is not positive-definite.  It turns out that there are no real solutions for $Y_{\alpha}$ and hence a solution consistent with the initial requirements does not exist.  
 		\\\\
 	For  the case where $Y$ is matrix valued such that  $Y_{\alpha}\tau^{\alpha}$ has positive eigenvalues boundedness of the exponential follows and therefore we have wedge-local fields. One example of such a matrix $Y$ is given by the following

 		\[ Y_1= \left( \begin{array}{ccc}
 		y_1 &0  \\
 		0 & y_1    \end{array} \right) ,\qquad Y_2=\left( \begin{array}{ccc}
 		y_2 &0  \\
 		0 & y_2   \end{array} \right)  ,\qquad Y_3=\left( \begin{array}{ccc}
 		y_3 &0  \\
 		0 & y_4   \end{array} \right),\] 
 		then the product $Y_{\alpha}\tau^{\alpha}$ is given by 
 		\[ Y_{\alpha}\tau^{\alpha}= \left( \begin{array}{ccc}
 		y_3 &y_1-iy_2  \\
 		y_1   + iy _2 & -y_4   \end{array} \right),\] 
 		which is hermitian for $y_1, y_2,y_3, y_4\in \mathbb{R}$. The eigenvalues of this matrix are real and given by 
 		\begin{align*}
 		\lambda_{1,2} = \frac{1}{2} \left(y_3-y_4\pm \sqrt{ 4 |y_1|^2 +4 |y_2|^2  +(y_3+y_4)^2} \right).
 		\end{align*}
 		By choosing $y_3=-y_4$ and   $y_1,y_2$ arbitrary we have for the eigenvalues of $Y_{\alpha}\tau^{\alpha}$,
 		\begin{align*}
 		\lambda_{1,2} =    y_3 \pm \sqrt{   |y_1|^2+ |y_2|^2  }.
 		\end{align*}
 		
 		For a specific choice of the constants,  namely   $y_3>\sqrt{|y_1|^2+|y_2|^2}$ the eigenvalues are strictly positive and hence wedge-locality follows for an entire family of  matrices.
 		
 	\end{proof} 
 	
 	\section{Conclusion and Outlook}
 	In this work we extended the warped convolutions formula, given in \cite{BS}, to an almost non-abelian version. The extended deformation formulas were further applied to the QM case and the outcome are gauge fields which are non-abelian. However, by calculating the non-abelian field strength tensor, where we used the deformed Heisenberg-Weyl algebra, the quadratic terms in the gauge fields vanish. Therefore, the deformation is dubbed almost non-abelian. Furthermore, we were able to construct an equivalent of the Moyal-Weyl plane in the non-abelian case. This was done by using similar arguments found in the abelian case  by using the Landau quantization, \cite{Muc1}. 
 	\\\\
 	We moved in the next step to QFT, i.e. we deformed the free scalar field with the extended version of warped convolutions. 
 	Moreover, we were able to relate the newly deformed fields with a QFT living on the non-abelian Moyal-Weyl plane. 
 	By an equivalent construction of \cite{GL1} we were able to relate the deformed fields to fields living on a wedge. Although wedge-covariance was proven the proof for wedge-locality was not possible if $Y$ was chosen to be vector-valued. However we could overcome the problem by choosing $Y$ to be matrix valued.  
 	\\\\
 	There are many possible extensions to this model. Let us mention one of the most intuitive extensions to this apparatus.  The constant vector $Y$ which was used in the definition of  deformation can be made operator dependent.    The changes concerning proofs are minimal; however this will have a greater effect on the physical side and will be in close relation to \cite{CCM}. Another possible extension is the deformation with  a spectral measure w.r.t. the non-abelian operator $U^{\tau}$. This is   work in progress, \cite{AM}. 
 	\\\\
 	The almost non-abelian deformation brings something new to the table, but it is isochronously not the last link of a chain. However, it can be considered as one further step   towards a non-abelian deformation theory.   
 	\section*{Acknowledgments} 
 	We would like to thank Prof. M. Rosenbaum for   insightful remarks. Furthermore, we are particularly grateful
 	for the crucial assistance given by   A. Andersson at various stages of this paper. We  would like 
 	to express our  great appreciation to Dr. Z. Much for an extensive proofreading.

 	\bibliographystyle{alpha}
 	\bibliography{allliterature1}

 \end{document}